\documentclass[10pt,journal]{IEEEtran}

\usepackage[dvipsnames]{xcolor}
\usepackage[normalem]{ulem}

\newcommand*{\tran}{\mathsf{T}}
\newcommand*{\herm}{\mathsf{H}}

\usepackage{cite}

\usepackage{graphicx}
\usepackage{epstopdf}
\graphicspath{{./figures/}}
\DeclareGraphicsExtensions{.pdf,.jpeg,.jpg,.png,.eps}

\usepackage{amsmath,amssymb,amsthm,amsfonts,amscd,amsbsy,amsxtra,amsfonts}
\usepackage{mathtools}
\usepackage{siunitx}
\usepackage{algorithm}
\usepackage{algpseudocode}
\usepackage{float}                
\floatstyle{ruled}\restylefloat{algorithm}
\algnewcommand\Input{\State\textbf{input: }}
\algnewcommand\Output{\State\textbf{output: }}
\algnewcommand\Initialize{\State\textbf{initialize: }}

\usepackage[shortlabels,inline]{enumitem}

\usepackage{booktabs}
\usepackage{multirow}

\usepackage[caption=false,font=footnotesize]{subfig}
\usepackage[export]{adjustbox}

\DeclareMathOperator{\sinc}{sinc}

\DeclareMathOperator*{\argmax}{arg\,max}

\DeclareMathOperator{\mse}{MSE}

\DeclarePairedDelimiter\floor{\lfloor}{\rfloor}
\DeclarePairedDelimiter\abs{\lvert}{\rvert}%
\DeclarePairedDelimiter\norm{\lVert}{\rVert}%

\newtheorem{lemma}{Lemma}
\newtheorem{corollary}{Corollary}

\newtheorem{remark}{Remark}
\newtheorem{proposition}{Proposition}

\IEEEoverridecommandlockouts

\begin{document}
\bstctlcite{IEEEexample:BSTcontrol}

\title{A Novel ISAC Waveform Based on Orthogonal Delay-Doppler Division Multiplexing with FMCW}

\author{
    Kehan~Huang,~\IEEEmembership{Graduate~Student~Member,~IEEE,}
    Akram~Shafie,~\IEEEmembership{Member,~IEEE,}
    Min~Qiu,~\IEEEmembership{Senior~Member,~IEEE}
    Elias~Aboutanios,~\IEEEmembership{Senior~Member,~IEEE,}
    and~Jinhong~Yuan,~\IEEEmembership{Fellow,~IEEE,}
    
    \thanks{The work of K. Huang, A. Shafie, M. Qiu, and J. Yuan was supported in part by the Australian Research Council (ARC) Discovery Project under Grant DP220103596, in part by the ARC Linkage Project under Grant LP200301482, and in part by the Connectivity Innovation Network, Australia. The work of M. Qiu was also supported in part by the SJTU ExploreX Funding under Grant SD6040004/153. This work has been presented in part at the 2025 IEEE International Conference on Communications \cite{Huang2025ODDMFMCW}. \emph{(Corresponding authors: Akram Shafie and Min Qiu.)}}
    \thanks{K. Huang, A. Shafie, E. Aboutanios, and J. Yuan are with the School of Electrical Engineering and Telecommunications (EET), University of New South Wales (UNSW), Sydney, NSW 2052, Australia (email: kehan.huang@unsw.edu.au; akram.shafie@unsw.edu.au; elias@ieee.org; j.yuan@unsw.edu.au). M. Qiu was with the School of EET, UNSW, Sydney, NSW 2052, Australia, when this work was conducted. He is now with the Global College, Shanghai Jiao Tong University, Shanghai 200240, China (email: min\_qiu@sjtu.edu.cn).}
}

\markboth{Accepted for publication in IEEE Transactions on Wireless Communications.}{Accepted for publication in IEEE Transactions on Wireless Communications.}

\maketitle

\begin{abstract}
    In this work, we propose the orthogonal delay-Doppler (DD) division multiplexing (ODDM) modulation with frequency modulated continuous wave (FMCW) (ODDM-FMCW) waveform to enable integrated sensing and communication (ISAC) with a low peak-to-average power ratio (PAPR). We first propose a square-root-Nyquist-filtered FMCW (SRN-FMCW) waveform to address limitations of conventional linear FMCW waveforms in ISAC systems. To better integrate with ODDM, we generate SRN-FMCW by embedding symbols in the DD domain, referred to as a DD-SRN-FMCW frame. A DD chirp compression receiver is designed to obtain the channel response efficiently. Next, we construct the proposed ODDM-FMCW waveform for ISAC by superimposing a DD-SRN-FMCW frame onto an ODDM data frame. A comprehensive performance analysis of the ODDM-FMCW waveform is presented, covering peak-to-average power ratio, spectrum, ambiguity function, and Cramér-Rao bound for delay and Doppler estimation. Numerical results show that the proposed ODDM-FMCW waveform delivers excellent ISAC performance in terms of root mean square error for sensing and bit error rate for communications.
\end{abstract}


\section{Introduction}\label{sec:intro}

Integrated sensing and communication (ISAC) has emerged as a key focus area in next-generation wireless networks, allowing more efficient use of frequency resources. Building upon a unified waveform, ISAC systems aim to support data transmission between communication nodes while simultaneously performing sensing tasks. With the potential to drive multiple next-generation applications, ISAC is especially relevant in high-mobility environments, including satellite and vehicle-to-everything (V2X) communications \cite{GP2024ISAC_6G,Wang2023DopplerSquint_OTFS,Wang2025FlexibleDDMA}.

To address the challenges of reliable communications in high-mobility environments, delay-Doppler (DD) modulation schemes, such as orthogonal time frequency space (OTFS) modulation \cite{Hadani2017OrthogonalModulation} and affine frequency division multiplexing (AFDM) modulation \cite{Bemani2021AFDM}, have been proposed to exploit channel diversity in the DD domain, offering enhanced robustness to doubly-selective channels. Building on OTFS, the recently proposed orthogonal DD division multiplexing (ODDM) modulation introduces a practical DD orthogonal pulse (DDOP) with low out-of-band emissions (OOBE) while maintaining \emph{sufficient} biorthogonality at fine DD resolution for the given signal bandwidth and time span \cite{Lin2022OrthogonalModulation,Tong2024ODDM_PhyChan}. Although initially designed for communications, the underlying DDOP shares similarities with pulse-Doppler radar signals, highlighting its potential for ISAC applications. Notably, ODDM and OTFS share the same DD domain information-bearing symbols, but the use of DDOP allows ODDM to directly embed DD domain symbols into continuous-time waveforms. The waveform generation of OTFS, on the other hand, is based on the time-frequency orthogonal pulse (TFOP), which does not align with the inherent operation of doubly selective channels \cite{Tong2024ODDM_PhyChan,Shafie2025ODDM_Spectrum_Orthogonality}. In this work, we investigate ISAC within the ODDM framework. However, the proposed ISAC design applies to both ODDM and OTFS by considering the corresponding input-output relations.

To enable sensing and channel estimation, a DD-domain impulse pilot (DDIP) was proposed for OTFS in \cite{Raviteja2019EmbeddedChannels}, and later adopted for ODDM in \cite{Tong2024ODDM_PhyChan,Huang2024ODDM_Performance,Shan2025ODDM_ChanEst_GridRefinement}. DDIP is well-localized in the fast-time/delay dimension, facilitating the effective decoupling of delay and Doppler parameters \cite{Richards2022Fundamentals_RadarSignalProcessing}. However, it can have a high peak-to-average power ratio (PAPR), leading to potential non-linear signal distortion and low power efficiency in practical transceivers.

Recent studies have designed DD-domain low-PAPR pilots by spreading the pilot energy across multiple pilot symbols. For instance, \cite{Liu2023ScatterPilot_OTFS} introduced the scattered pilots, where the placement of pilot symbols is carefully chosen to avoid self-interference. However, this placement relies on the frame size and channel support, making it less robust to practical systems. To mitigate self-interference, orthogonal sequences, such as Zadoff-Chu (ZC) sequence and chirp sequence, have been employed \cite{Shi2021MIMO-OTFS_DetPilot,Shen2025SC_DDEstimation,Ma2025OTFS_LowPAPRPilots}. In \cite{Shi2021MIMO-OTFS_DetPilot}, a ZC-based spread pilot was proposed, which requires an excessive guard interval. This is because a cyclic extension is padded to utilize the cyclic orthogonality of ZC sequences, and the length of the sequence necessitates the span of the guard interval to prevent interference with data symbols. Similarly, \cite{Ma2025OTFS_LowPAPRPilots} proposed another ZC-based pilot without guard interval. Despite higher spectral efficiency, this design fails to fully exploit the cyclic orthogonality, which may compromise sensing performance. To further improve spectral efficiency, \cite{Mishra2022OTFS_SuperimposedPilots} introduced a superimposed random pilot and recovered the channel state using a minimum mean square error (MMSE) estimator. However, this approach requires prior knowledge of the channel support, which is often unavailable for communication receivers.

While \cite{Shi2021MIMO-OTFS_DetPilot} and \cite{Ma2025OTFS_LowPAPRPilots} considered fractional Doppler shifts, the aforementioned pilot designs do not address the off-grid sampling of both delay and Doppler shifts---an inherent aspect of ISAC systems. To account for this, \cite{Bondre2024DFRC-OTFS_RootMUSIC} proposed a chirp-based spread pilot and applied a root-MUSIC algorithm for DD estimation, but the ideal pulse assumed for OTFS is not realizable. Building on practical pulses, \cite{Ubadah2024ZakOTFS_ISAC} proposed a 2D-chirp-based spread pilot. Nonetheless, the 2D chirp sequence lacks orthogonality, which can limit its sensing performance. Therefore, designing a channel-robust, low-PAPR ISAC waveform remains an open problem.

Frequency modulated continuous wave (FMCW) is widely used in automotive radars for its simplicity and low PAPR \cite{Richards2022Fundamentals_RadarSignalProcessing,Patole2017AutomotiveRad_Review}. Hence, the combination of ODDM and FMCW can potentially produce low-PAPR ISAC waveforms that are fully compatible with OTFS/ODDM technologies. \cite{Zegrar2022OTFS_FMCW} investigated the symbol-level representation of FMCW signals within the OTFS framework. However, considering practical continuous-time signals, the combination of OTFS/ODDM and FMCW has not been well researched due to the following challenges. Firstly, FMCW signals do not satisfy the Nyquist intersymbol interference (ISI) criterion with respect to OTFS/ODDM signals, which can cause complicated mutual interference in ISAC systems \cite{Zegrar2022OTFS_FMCW,Tong2024ODDM_PhyChan,Zhang2025OCDM_FDE}. Secondly, with limited bandwidth and a high pulse repetition interval, the stretch processing technique used in conventional FMCW radar receivers is susceptible to range-Doppler coupling and range skew \cite{Richards2022Fundamentals_RadarSignalProcessing}. Finally, FMCW signals exhibit high OOBE, making them unsuitable for practical communication systems.

To address these challenges, we propose an ISAC framework based on a novel ODDM-FMCW waveform. The main contributions of the paper are as follows:

\textbf{1}) We propose the square-root Nyquist (SRN)-filtered FMCW (SRN-FMCW) waveform to enable sensing and channel estimation with controlled chirp compression sidelobes and reduced transmission overhead. To support coexistence with ODDM systems, we introduce the DD-domain embedded SRN-FMCW (DD-SRN-FMCW) waveform, which can be generated using an ODDM transmitter. Building on the ODDM receiver, we further propose \emph{DD chirp compression} to efficiently obtain the DD response (DDR) of the channel.

\textbf{2}) We propose the ODDM-FMCW waveform for ISAC, where an ODDM data frame is transmitted with a superimposed DD-SRN-FMCW frame. For data-aided sensing (DAS) at the collocated sensing receiver, we introduce a super-resolution sensing algorithm based on orthogonal matching pursuit (OMP) \cite{Rasheed2020OTFS_MU_OMP}. For joint channel estimation and data detection (JCEDD) at the communication receiver, we combine OMP and soft successive interference cancellation with minimum mean square error (SIC-MMSE) detector \cite{Li2024SICMMSE_Turbo}.

\textbf{3}) We present a comprehensive performance analysis for the proposed DD-SRN-FMCW and ODDM-FMCW signals. First, a good approximation is given for their PAPR complementary cumulative distribution function (CCDF). Second, we characterize their power spectral density (PSD) and highlight the differences from a linear FMCW signal. Additionally, the ambiguity function of the DD-SRN-FMCW signal and the Cramér-Rao bound (CRB) for delay and Doppler estimation are derived, confirming its strong sensing capability.

\textbf{4}) We demonstrate via simulations that the proposed ODDM-FMCW waveform offers significantly lower PAPR than ODDM with a superimposed DDIP (ODDM-DDIP). Owing to its robustness against channel estimation errors, ODDM-FMCW consistently outperforms ODDM-DDIP in terms of bit error rate (BER) at the communication receiver, approaching that of ODDM with perfect channel state information (CSI). Moreover, low normalized root mean square error (NRMSE) in delay and Doppler estimation can be achieved at the sensing receiver. 

\textbf{Notations}: $(\cdot)^*$ and $(\cdot)^\herm$ denote the complex conjugate and Hermitian transpose, respectively. $\left[\cdot\right]^\dagger$ denotes the pseudoinverse. $\boldsymbol{F}_N$ denotes the normalized $N$-point DFT matrix. $[\cdot]_N$ means modulo $N$. $\floor{\cdot}$ means floor. $\delta(\cdot)$ denotes the Dirac delta function. $\norm{\boldsymbol{A}}$ outputs the 2-norm of matrix $\boldsymbol{A}$. Denote the column-stacking vectorization of $\boldsymbol{A}\in\mathbb{C}^{m\times n}$ by $\boldsymbol{a}=\mathrm{vec}(\boldsymbol{A})$, with inverse operation $\mathrm{vec}_{m\times n}^{-1}(\boldsymbol{a})=\boldsymbol{\boldsymbol{A}}$. $\gcd(\cdot)$ outputs the greatest common divisor. The important symbols and their definitions are summarized in Table \ref{tab:notation}.

\begin{table}[t]
  \caption{Important symbols used in this paper.}
  \label{tab:notation}
  \centering
  \begin{tabular}{ll}
    \toprule
    Symbol & Meaning \\ 
    \midrule
    $M$, $N$                    & Number of delay and Doppler bins, respectively.\\
    $a(t)$, $g(t)$              & SRN pulse and Nyquist pulse.\\
    $\epsilon$                  & Chirp rate.\\
    $c(t)$, $c_a(t)$            & Linear chirp and SRN-filtered chirp.\\
    $\boldsymbol{c}$            & Discrete chirp sequence.\\
    ${}_{ex}$, ${}_{ce}$        & Linearly and cyclically extended signal, respectively.\\
    $\boldsymbol{X}$            & General ODDM frame or ODDM-FMCW frame.\\
    $\boldsymbol{X}_c$, $\boldsymbol{X}_d$    & DD-SRN-FMCW frame and ODDM data frame.\\
    $\rho$                      & Chirp-data-power-ratio (CDPR).\\
    \bottomrule
  \end{tabular}
\end{table}

\section{Preliminaries and Motivations}\label{sec:preliminaries&motivations}

\begin{figure*}[t]
    \centering
    \includegraphics[width=2\columnwidth,trim={0 0 0 0},clip]{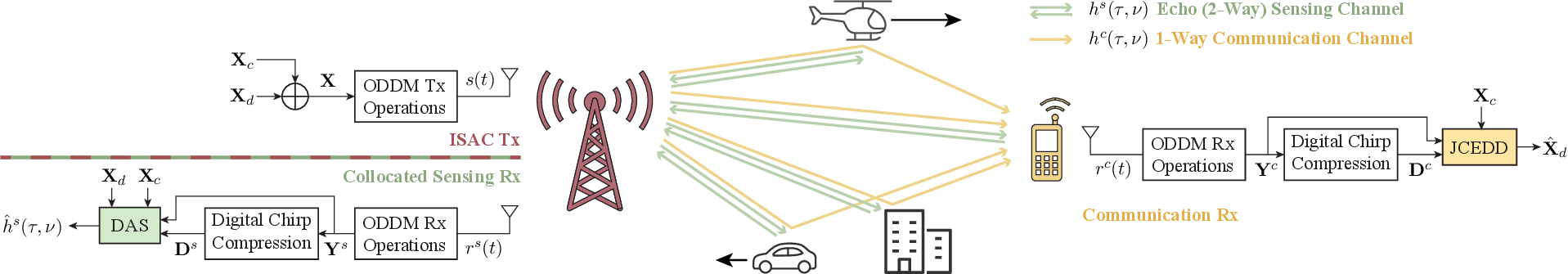}
    \caption{Considered ISAC scenario.}
    \label{fig:isac_chan}
\end{figure*}

In this work, we explore ISAC in a dynamic environment, as depicted in Fig.~\ref{fig:isac_chan}. We consider an ISAC transmitter that transmits the proposed ODDM-FMCW signal to a communication receiver, while simultaneously using the same signal for monostatic sensing with a collocated sensing receiver. The ODDM-FMCW frame $\boldsymbol{X}$ is constructed by superimposing a DD-SRN-FMCW frame $\boldsymbol{X}_c$ onto an ODDM data frame $\boldsymbol{X}_d$, where $\boldsymbol{X}_c$ functions as both the sensing signal and channel estimation pilot. An ODDM transmitter is used to transmit $\boldsymbol{X}$. At both receivers, the time-domain received signal $r(t)$ is first processed by the ODDM receiver to obtain the received ODDM-FMCW frame $\boldsymbol{Y}$. Then, DD chirp compression is performed to compute the DDR $\boldsymbol{D}$ of the channel. For the collocated sensing receiver, a DAS algorithm is used to enable super-resolution sensing. For the communication receiver, a JCEDD algorithm is used to perform data detection.

In the rest of this section, we present preliminaries for the ODDM-FMCW system introduced in the following sections. We begin with the doubly-selective channel. Next, we detail the simplified/approximate implementation of ODDM modulation \cite{Lin2022OrthogonalModulation}. We then review the linear FMCW radar waveform and its stretch-processing receiver. Additionally, we discuss chirp compression, a receiver technique employed in chirp-based pulsed radars. Finally, we highlight challenges in applying these standalone ODDM and FMCW systems to ISAC.

\subsection{Doubly-Selective LTV Channel}\label{sec:ds_channel}

To characterize signal propagation in Fig.~\ref{fig:isac_chan}, we adopt a doubly-selective linear time-varying (LTV) channel model with $P$ point scatterers. Each scatterer has complex path coefficient $\alpha_p$, initial propagation delay $\tau_p$, and velocity $v_p$. We define $v_p$ as the changing rate of the propagation distance along the $p$-th path, yielding a time-varying propagation delay $\dot{\tau}_p(t) = \tau_{p} + \frac{v_p}{c}t$, with $c$ being the speed of light. Following the Swerling 1 target model and short-time assumption \cite{Richards2022Fundamentals_RadarSignalProcessing,Tong2024ODDM_PhyChan}, $\alpha_p$ and $v_p$ remain constant during a radar coherent processing interval, which is also the ODDM frame duration.

Suppose the transmitted \emph{passband} signal $s_\mathrm{pb}(t)$ occupies a bandwidth $\mathcal{B}$ around carrier frequency $f_c$. When $\mathcal{B}\ll f_c$, the received \emph{passband} signal can be expressed as\footnote{Although passband signals are real-valued, without loss of generality, we represent them as analytic (complex) signals to streamline the derivations in this work. This is justified under the narrowband assumption $\mathcal{B}\ll f_c$.}
\begin{align}\label{eq:r_pb}
    r_\mathrm{pb}(t)=\sum_{p=1}^{P} \alpha_p s_{\mathrm{pb}}(t-\dot{\tau}_p(t)) + z_{\mathrm{pb}}(t),
\end{align}
where $z_{\mathrm{pb}}(t)$ is the passband Gaussian noise.

Let $s_{\mathrm{pb}}(t)=s(t)e^{j2\pi f_ct}$ denote the passband of a linearly modulated baseband signal $s(t)$. Substituting into $\eqref{eq:r_pb}$ gives
\begin{align}\label{eq:r_pb_linear_mod}
    r_{\mathrm{pb}}(t) = \sum_{p=1}^{P} \alpha_p s(t\!-\!\dot{\tau}_p(t))e^{j2\pi f_c(t-\tau_p)}e^{j2\pi\nu_pt} \!+\! z_{\mathrm{pb}}(t),
\end{align}
where $\nu_p=-\frac{f_cv_p}{c}$ is the Doppler shift of the $p$-th path. Assuming perfect carrier synchronization and the narrowband approximation $s(t-\dot{\tau}_p(t))\approx s(t-\tau_p)$, down-conversion of $r_{\mathrm{pb}}(t)$ by $e^{-j2\pi f_ct}$ yields
\begin{align}\label{eq:s2r_baseband}
    r(t) =\sum_{p=1}^{P} h_p s(t-\tau_p) e^{j2\pi \nu_pt} + z(t),
\end{align}
where
\begin{align}\label{eq:hp}
    h_p = \alpha_p e^{-j2\pi f_c\tau_p}
\end{align}
is the baseband channel coefficient of the $p$-th path, and $z(t)$ is the baseband complex Gaussian noise. We clarify that \eqref{eq:s2r_baseband} takes the form of Bello's baseband LTV channel model \cite{Bello1963CharacterizationChannels}, characterized by a DD-domain spreading function:
\begin{align}\label{eq:dd_chan}
    h(\tau,\nu)=\sum_{p=1}^{P} h_p \delta (\tau-\tau_p)\delta (\nu-\nu_p).
\end{align}
In our ISAC scenario illustrated in Fig.~\ref{fig:isac_chan}, two doubly-selective channels are of interest: the echo sensing channel $h^s(\tau,\nu)$ and the one-way communication channel $h^c(\tau,\nu)$.

\subsection{ODDM Modulation}\label{sec:oddm}

ODDM was proposed for doubly-selective channels as a DD-domain multicarrier modulation. In this section, we detail the simplified/approximate ODDM implementation \cite{Lin2022OrthogonalModulation}.

\subsubsection{ODDM Transmitter}\label{sec:oddm_tx}

Consider the transmission of $MN$ symbols within nominal bandwidth $\mathcal{B}=\frac{M}{T}$ and transmission time $NT$. ODDM arranges the symbols $\boldsymbol{X}\in\mathbb{C}^{M\times N}$ onto a DD grid with delay resolution $\Delta \tau=\frac{T}{M}$ and Doppler resolution $\Delta \nu=\frac{1}{NT}$. At the transmitter, $\boldsymbol{X}$ is first converted into the delay-time (DT) domain by $N$-point IDFT to obtain $\boldsymbol{X}^{\mathrm{DT}} = \boldsymbol{X}\boldsymbol{F}_N^\herm$. Then, $\boldsymbol{X}^{\mathrm{DT}}$ is vectorized to obtain the time-domain digital samples as $\boldsymbol{s} = \mathrm{vec}(\boldsymbol{X}^{\mathrm{DT}})$, with its $q$-th element denoted by $s[q]$. Finally, the continuous-time baseband signal $s(t)$ is generated by $a(t)$-based pulse shaping \cite{Lin2022OrthogonalModulation}
\begin{align}
    s(t) &= \sum_{q=0}^{MN-1}s[q]a\left(t-q\frac{T}{M}\right),
    \label{eq:s_a}
\end{align}
where $a(t)$ is a square-root Nyquist (SRN) pulse that serves as the subpulse of ODDM. To accommodate underspread channels with maximum delay spread $\tau_{\max}<T$, ODDM adopts a frame-wise cyclic prefix (CP) with a duration of $\tau_{\max}$.

\subsubsection{ODDM Receiver}\label{sec:oddm_rx}
At the ODDM receiver, the received baseband signal $r(t)$ is matched filtered using $a(t)$ and sampled at $t=q\frac{T}{M}$ for $q\in\{0,\dots,MN-1\}$, to obtain the time-domain received sample vector $\boldsymbol{r}\in\mathbb{C}^{MN\times1}$. For the LTV channel in \eqref{eq:dd_chan}, the recevied samples $r[q]$ become \cite{Tong2024ODDM_PhyChan}
\begin{align}
    r[q] = \sum_{p=1}^{P} h_p \Biggl(\sum_{d=0}^{2Q} &e^{j2\pi\frac{k_p\left(q-l_p-d\right)}{MN}} g\left(\left(d+\floor{l_p}-l_p\right)\frac{T}{M}\right) 
    \nonumber\\
    &\times  s\left[\left[q-\floor{l_p}-d\right]_{MN}\right]\Biggr) + z[q],
    \label{eq:io_time}
\end{align}
where $l_p=\tau_p\frac{M}{T}$ and $k_p=\nu_pNT$ are the normalized delay and Doppler shifts of the $p$-th path, respectively, $g(t) = a(t)\circledast a^*(t)$ is a causal symmetric Nyquist pulse centered at $t=Q\frac{T}{M}$, and $z[q]\sim\mathcal{CN}\left(0,\sigma_z^2\right)$ is the sampled noise. $Q$ is the half-span truncation length of $g(t)$ (in symbol intervals) for practical implementation. We set $Q=20$ so that the pulse support beyond $2Q\frac{T}{M}$ is negligible \cite{Tong2024ODDM_PhyChan,Shafie2024DDOP_TF_Loc,Lin2022OrthogonalModulation}. Note that $l_p$ and $k_p$ can take fractional (off-grid) values.

To convert time-domain samples $\boldsymbol{r}$ back to the DD domain, the DT domain matrix is constructed as $\boldsymbol{Y}^{\mathrm{DT}}=\mathrm{vec}_{M\times N}^{-1}(\boldsymbol{r})$. Then, $N$-point DFT is performed to obtain the received DD-domain matrix as $\boldsymbol{Y} = \boldsymbol{Y}^{\mathrm{DT}}\boldsymbol{F}_N$. The symbol at the $m$-th delay and $n$-th Doppler bin in $\boldsymbol{Y}$ is given by \cite{Tong2024ODDM_PhyChan}
\begin{align}\label{eq:io_dd}
    Y[m,n] = \sum_{p=1}^{P} Y_p[m,n] + Z[m,n],
\end{align}
where
\begin{align}
    Y_p[m&,n] = h_p \sum_{d=0}^{2Q} e^{j2\pi\frac{\left(m-l_p-d\right)k_p}{MN}} g\left(\left(d\!+\!\floor{l_p}\!-\!l_p\right)\frac{T}{M}\right)
    \nonumber\\
    &\times\sum_{\tilde{n}=0}^{N-1} \phi(\tilde{n}\!+\!k_p\!-\!n) \psi[m,d,\tilde{n}] X\!\left[\left[m\!-\!\floor{l_p}\!-\!d\right]_{M},\tilde{n}\right]
    \nonumber
\end{align}
is the signal component corresponding to the $p$-th path, and $Z[m,n]\sim\mathcal{CN}(0,\sigma_z^2)$ is the DD-domain sampled noise. Here,
\begin{align}\label{eq:phi}
    \phi(n) = \frac{1}{N}\sum_{n'=0}^{N-1}e^{j2\pi\frac{nn'}{N}} = \frac{1-e^{j2\pi n}}{N\left(1-e^{j\frac{2\pi}{N}n}\right)}
\end{align}
denotes the Dirichlet kernel, whereas
\begin{align}\nonumber
    \psi[m,d,\tilde{n}] = 
    \begin{cases}
        e^{-j2\pi\frac{\tilde{n}}{N}}, & m\!-\!\floor{l_p}\!-\!d < 0,
        \\
        1, & 0\leq m\!-\!\floor{l_p}\!-\!d \leq M\!-\!1,
        \\
        e^{j2\pi\frac{\tilde{n}}{N}}, & m\!-\!\floor{l_p}\!-\!d > M\!-\!1,
    \end{cases}
\end{align}
represents the CP-induced phase rotation. Leveraging the input-output relation in \eqref{eq:io_dd}, DDIP has become the most considered pilot design for channel estimation in OTFS/ODDM literature \cite{Raviteja2019EmbeddedChannels,Tong2024ODDM_PhyChan,Huang2024ODDM_Performance,Shan2025ODDM_ChanEst_GridRefinement,Wang2023DopplerSquint_OTFS,Yuan2021DataAided_ChanEst}.

\subsection{FMCW Radar and Stretch Processing}\label{sec:fmcw_stretchproc}
FMCW radar is widely used in automotive applications due to its low complexity and energy efficiency \cite{Richards2022Fundamentals_RadarSignalProcessing,Patole2017AutomotiveRad_Review}. In this section, we review its signal model and highlight the key differences from conventional communication systems.

Consider a passband linear FMCW signal with $N$ chirps $c(t)$ and a pulse repetition interval (PRI) of $T$, given by
\begin{align}\label{eq:sc}
    s_c(t) = \sum_{n=0}^{N-1}c(t-nT).
\end{align}
Each chirp $c(t)$ is defined as
\begin{align}\label{eq:c_t}
    c(t) = \dot{c}(t)\Pi_T\!\left(t-\frac{T}{2}\right),
\end{align}
where
\begin{align}\label{eq:cdot_t}
    \dot{c}(t) = e^{j2\pi\left(f_c+\frac{\epsilon}{2}t\right)t}
\end{align}
is an infinite-time linear chirp and $\Pi_T(t)$ denotes a rectangular window of duration $T$. Here $f_c$ is the carrier frequency and $\epsilon$ is the chirp rate. $c(t)$ has an instantaneous frequency of $\frac{d}{dt}\arg(c(t))=f_c+\epsilon t$ and a total bandwidth of $\mathcal{B}=\epsilon T$. As will be shown in Section~\ref{sec:spectrum}, $s_c(t)$ exhibits distinct spectral characteristics compared to communication signals.

When considering the passband channel described in \eqref{eq:r_pb} in Section~\ref{sec:ds_channel}, the received FMCW signal becomes
\begin{align}\label{eq:rc}
    r_c(t) = \sum_{p=1}^{P} \alpha_p \!\sum_{n}^{N-1}c(t\!-nT\!-\tau_p) e^{j2\pi\nu_pt} + z_{\mathrm{pb}}(t).
\end{align}

To obtain the channel's DDR from $r_c(t)$, the FMCW radar receiver employs \emph{stretch processing} (also known as \emph{dechirping}) \cite{Barton2004RadSys_Ana&Mod,Richards2022Fundamentals_RadarSignalProcessing,Patole2017AutomotiveRad_Review}. In particular, $r_c(t)$ is mixed with $s_c(t)$ to obtain the \emph{beat frequency} signal $r_b(t) = s_c(t)r_c^*(t)$. Then, $r_b(t)$ is sampled with interval $\frac{T}{M}$ and reshaped to form a \emph{radar data matrix} $\boldsymbol{R}_b\in\mathbb{C}^{M\times N}$, given by
\begin{align}
    R_b[&m,n] = r_b\left(nT+m\frac{T}{M}\right)
    \nonumber\\
    &=
    \begin{cases}
        \sum_{p=1}^{P} h_p' R_{b,p}[m,n], & m\frac{T}{M} \geq \tau_p,\\
        \sum_{p=1}^{P} h_p' R_{b,p}[m,n]e^{j2\pi\epsilon T\tau_p}, & \text{otherwise},
    \end{cases}
    \label{eq:Rb}
\end{align}
where
\begin{align}\label{eq:hp_fmcw}
    h_p'=\alpha_p e^{j2\pi\tau_p(f_c-\frac{\epsilon}{2}\tau_p)}
\end{align}
is the channel coefficient for the $p$-th path and
\begin{align}\label{eq:Rbp}
    R_{b,p}[m,n] &= \underbrace{e^{j2\pi\frac{T}{M}(\epsilon\tau_p-\nu_p)m}}_{\text{fast-time varying}} \underbrace{e^{-j2\pi T\nu_p n}}_{\text{slow-time varying}}
\end{align}
is the corresponding fast-slow time response.\footnote{
    ``fast time" and ``slow time" follow radar terminology, where ``fast time" refers to samples within a PRI and ``slow time" refers to variations across PRIs \cite{Richards2022Fundamentals_RadarSignalProcessing}. They align with the concepts of delay and time in OTFS/ODDM.
}
Finally, a 2D-DFT yields the channel DDR as $\boldsymbol{D}_b=\boldsymbol{F}_M\boldsymbol{R}_b\boldsymbol{F}_N$.

By inspecting \eqref{eq:Rbp}, it can be inferred that the fast-time and slow-time components of $R_{b,p}[m,n]$ will appear as separable sinc profiles centered at $\tau_p-\frac{\nu_p}{\epsilon}$ and $\nu_p$ in $\boldsymbol{D}_b$, respectively. The slow-time sinc directly corresponds to the Doppler shift $\nu_p$. However, the fast-time sinc gives $\tau_p-\frac{\nu_p}{\epsilon}$, which is the true delay $\tau_p$ distorted by the Doppler shift. This phenomenon is known as \emph{range-Doppler coupling} \cite{Barton2004RadSys_Ana&Mod,Richards2022Fundamentals_RadarSignalProcessing}. In addition, the extra term $e^{j2\pi\epsilon T\tau_p}$ in the second case of \eqref{eq:Rb} introduces a discontinuity in phase, giving rise to the \emph{range skew effect} \cite{Richards2022Fundamentals_RadarSignalProcessing}. These issues become more pronounced when the allocated time ($T$) and bandwidth ($\mathcal{B}$) are relatively limited, thereby restricting FMCW radar performance under typical communication constraints.

From an ISAC standpoint, communication receivers require the baseband channel coefficient $h_p$ as in \eqref{eq:hp} for symbol detection, whose phase depends linearly on $\tau_p$. In contrast, the FMCW channel coefficient $h_p'$ in \eqref{eq:hp_fmcw} contains a distinct phase term proportional to $\tau_p^2$. Accordingly, using FMCW radar for baseband channel estimation requires a nonlinear conversion from $h_p'$ to $h_p$, which can amplify channel estimation errors.

\subsection{Chirp-based Pulsed Radar and Chirp Compression}\label{sec:chirp_compression}

Pulse compression is a common receiver technique for pulsed radars. It is also termed chirp compression when chirp signals are used as radar pulses \cite{Barton2004RadSys_Ana&Mod}. We explore chirp compression in this work since it shares similarities with the matched filtering operations in OTFS/ODDM and single-carrier (SC) system receivers.

Chirp compression relies on the fact that the autocorrelation of the chirp signal $c(t)$ in \eqref{eq:c_t} yields a sinc-like pulse. To understand this, we first look into the ideal chirp compression output, given by the crosscorrelation between the infinite-time chirp $\dot{c}(t)$ in \eqref{eq:cdot_t} and the finite-time chirp $c(t)$ in \eqref{eq:c_t}
\begin{align}
    A_{\dot{c},c}(\tau)&\triangleq\int_{-\infty}^\infty \dot{c}(t) c^*(t-\tau)dt=e^{j2\pi(f_c-\frac{\epsilon}{2}\tau)\tau}\int_{0}^T e^{j2\pi\epsilon\tau t}dt
    \nonumber\\
    &=T e^{j2\pi(f_c-\frac{\epsilon}{2}\tau)\tau} e^{j\pi\epsilon\tau T}\sinc(\epsilon T\tau),
    \label{eq:Act_ideal}
\end{align}
which is essentially the ideal delay ambiguity function of a chirp signal. Since \eqref{eq:Act_ideal} is a sinc-type function of delay, its sharp mainlobe admits straightforward delay estimation.

\begin{figure}[t]
    \centering
    \subfloat[$T\mathcal{B}=64$.]{\label{fig:chirp_compression_m64}\includegraphics[width=1\columnwidth,trim={0 0 0 0},clip]{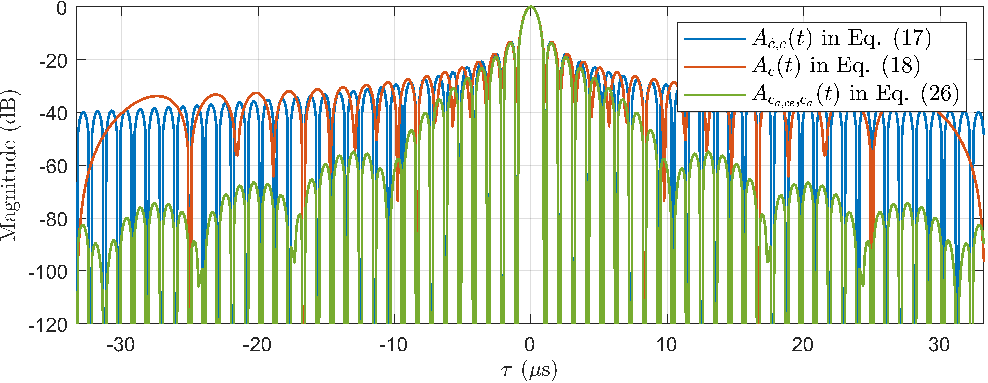}}
    \\
    \centering
    \subfloat[$T\mathcal{B}=256$.]{\label{fig:chirp_compression_m256}\includegraphics[width=1\columnwidth,trim={0 0 0 0},clip]{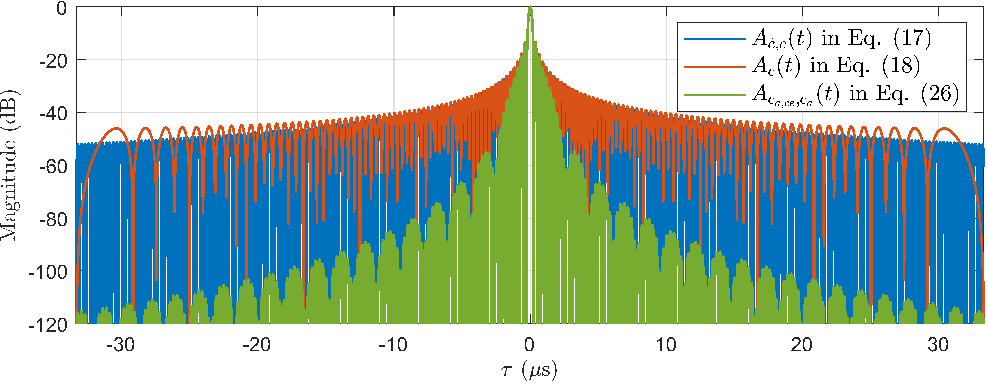}}
    \caption{Delay ambiguity (chirp compression output) of chirps.}
    \label{fig:chirp_compression}
\end{figure}

However, $A_{\dot{c},c}(\tau)$ assumes an ideal infinite-time chirp $\dot{c}(t)$. In practice, finite-time chirps suffer pulse degradation, especially with a limited time-bandwidth product $T\mathcal{B}$. To show this, we now consider the autocorrelation of $c(t)$, written as
\begin{align}
    A_c(\tau)&\triangleq\int_{-\infty}^\infty c(t) c^*(t-\tau)dt
    \nonumber\\
    &=e^{j2\pi(f_c-\frac{\epsilon}{2}\tau)\tau} \int_{\max(0,\tau)}^{\min(T,T+\tau)} e^{j2\pi\epsilon\tau t}dt.
    \label{eq:Act_practical}
\end{align}
Fig.~\ref{fig:chirp_compression} compares $A_{\dot{c},c}(\tau)$ (blue curve) and $A_c(\tau)$ (orange curve) using the parameters in Section~\ref{sec:numerical} (the green curve will be discussed later in Section~\ref{sec:rrc_chirp}). Compared to the ideal $A_{\dot{c},c}(\tau)$, $A_c(\tau)$ exhibits an elevated, broadened sidelobe floor, which becomes more pronounced for lower $T\mathcal{B}$. This directly follows from the $\tau$-dependent integration bounds in \eqref{eq:Act_practical} and, in the frequency domain, corresponds to the classical \emph{Fresnel ripple} \cite{Cook1964ChirpCompression_Sidelobe}. In addition, $A_c(\tau)$ lacks a closed-form expression, hindering superresolution sensing. Most importantly, the zeros of $A_c(\tau)$ do not align with those of the sinc function. That is, the Nyquist ISI criterion may not be satisfied, which complicates the mutual interference between sensing and communication signals in an ISAC context.

\subsection{Challenges and Motivations}\label{sec:challenges&motivations}

The ODDM system discussed in Section~\ref{sec:oddm} is designed for data transmission over doubly-selective channels. Although the widely adopted DDIP shows promising channel estimation and sensing performance \cite{Raviteja2019EmbeddedChannels,Shan2025ODDM_ChanEst_GridRefinement,Yuan2021DataAided_ChanEst}, it exhibits high PAPR and low energy efficiency. Conversely, while the FMCW waveform in Section~\ref{sec:fmcw_stretchproc} provides unit-PAPR sensing capability, it faces challenges in the integration with communication systems, including range-Doppler coupling, range skew, and error amplification. Moreover, passband FMCW systems may not be directly compatible with baseband modulation systems.

The chirp compression method described in Section~\ref{sec:chirp_compression} reveals a potential solution to these challenges. In particular, chirp compression treats each chirp as an individual radar pulse and therefore better aligns with pulse-based modulation schemes such as OTFS/ODDM and SC. This, however, introduces new difficulties, including pronounced sidelobes, a lack of closed-form expression, and violation of the Nyquist ISI criterion. These limitations motivate the DD-SRN-FMCW sensing waveform we propose in the next section.

\section{From Linear FMCW to DD-SRN-FMCW}\label{sec:linearfmcw2rrcfmcw}

In this section, we propose the DD-SRN-FMCW waveform to overcome the limitations of linear FMCW in ISAC applications highlighted in Section~\ref{sec:challenges&motivations}. We first design a SRN-filtered chirp, followed by the SRN-FMCW waveform incorporating multiple chirps. Finally, we introduce DD-SRN-FMCW to achieve compatibility with the ODDM framework.

\subsection{SRN-Filtered Chirp with Chirp Compression}\label{sec:rrc_chirp}

As discussed in \ref{sec:chirp_compression}, the sinc-type delay ambiguity function $A_{\dot{c},c}(\tau)$ in \eqref{eq:Act_ideal} is not achievable with the practical finite-time chirp $c(t)$. However, we note that the sinc-type delay ambiguity function can be \emph{locally} achieved by transmitting an extended chirp signal, e.g.,
\begin{align}\label{eq:c_ex}
    c_{ex}(t)=
    \begin{cases}
        \dot{c}(t), & t\in(-\tau_{\max},T+\tau_{\max}), \\
        0 , & \text{otherwise}.
    \end{cases}
\end{align}
Then, the crosscorrelation between $c_{ex}(t)$ and $c(t)$, given by $A_{c_{ex},c}(\tau) = \int_{-\infty}^\infty c_{ex}(t)c^*(t-\tau)dt$, satisfies $A_{c_{ex},c}(\tau)=A_{\dot{c},c}(\tau)$ for delays $\tau\in(-\tau_{\max},\tau_{\max})$. That is, Fresnel-ripple sidelobes can be avoided within a finite delay interval.

However, $c_{ex}(t)$ extends in both time and frequency compared to $c(t)$, taking more spectral resources to support dispersive channels with large delay spreads. Therefore, we now leverage the sampling theorem to reduce the transmission overhead. To begin with, we sample $\dot{c}(t)$ at $t=m\frac{T}{M}$ for $m\in\mathbb{Z}$ to obtain the discrete chirp sequence
\begin{align}\label{eq:cdot_m}
    \dot{c}[m] = e^{j2\pi\left(\frac{f_cT}{M}m+\frac{\epsilon T^2}{2M^2}m^2\right)}.
\end{align}
Inspired by the sinc-type delay ambiguity in \eqref{eq:Act_ideal}, we establish the following Lemma for $\dot{c}[m]$:
\begin{lemma}\label{lmm:zla}
    If $\frac{\epsilon T^2}{M}\in\mathbb{Z}$ and $\gcd\left(\frac{\epsilon T^2}{M},M\right)=1$, then $\dot{c}[m]$ satisfies the zero linear autocorrelation (ZLA) property:
    \begin{align}\label{eq:zla}
        \sum_{m=0}^{M-1} \dot{c}[m]\dot{c}^{*}[m+\zeta]=M\delta[\zeta],
    \end{align}
    for $\zeta\in\mathbb{Z}$.
\end{lemma}
\begin{IEEEproof}
    The proof is given in Appendix \ref{app:proof_zla}.
\end{IEEEproof}

Given the ZLA property in Lemma \ref{lmm:zla}, the requirement of linear extension of the linear chirp $\dot{c}(t)$ can be relaxed to the linear extension of the discrete chirp sequence $\dot{c}[m]$, which avoids the transmission overhead in bandwidth. Based on Lemma \ref{lmm:zla}, we further relax the requirement of linear extension to cyclic extension by establishing the following proposition.
\begin{proposition}\label{prop:zca}
    When the three conditions 
    \begin{enumerate*}[(a),itemjoin={{, }},itemjoin*={{, and }}]
        \item $\frac{\epsilon T^2}{M}\in\mathbb{Z}$
        \label{enu:cond1}
        \item $\gcd\left(\frac{\epsilon T^2}{M},M\right)=1$
        \label{enu:cond2}
        \item $f_cT+\frac{\epsilon T^2}{2}\in \mathbb{Z}$
        \label{enu:cond3}
    \end{enumerate*}
    hold, the discrete chirp sequence $c[m]=\dot{c}[m],~\forall m\in\{0,\dots,M-1\}$ satisfies the zero cyclic autocorrelation (ZCA) property:
    \begin{align}\label{eq:zca}
        \sum_{m=0}^{M-1}c[m]c^{*}\bigl[[m+\zeta]_{M}\bigr]=M\delta\bigl[\left[\zeta\right]_M\bigr],
    \end{align}
    for $\zeta\in\mathbb{Z}$.
\end{proposition}
\begin{IEEEproof}
    The proof is given in Appendix \ref{app:proof_zca}.
\end{IEEEproof}
Proposition~\ref{prop:zca} indicates that the discrete chirp sequence $c[m]$ can satisfy the ZCA property in \eqref{eq:zca} by attaching a cyclic prefix (CP) and a cyclic suffix (CS). While this section considers a single chirp, the restriction from linear to cyclic extension is most useful for a chirp train, where each chirp naturally reuses its neighbors as the CP and CS. The next subsection constructs a multi-chirp waveform that exploits the ZCA property to reduce the transmission overhead in time.

\begin{remark}
    Most communication systems can be described by convolution operations, where the channel response and modulation pulse are strictly causal, making it possible to omit CS. Chirp compression, however, relies on correlation, which is a non-causal operation. Therefore, a CS of length $M$ is required to accommodate such non-causality.
\end{remark}

By inspecting the conditions in Proposition \ref{prop:zca}, we can get a trivial set of chirp parameters for all even $M$, given by $f_c=0$ and $\epsilon=\frac{M}{T^2}$. Interestingly, if we substitute these parameters back into the linear FMCW signal $s_c(t)$ in \eqref{eq:sc}, it will result in a delay resolution of $\Delta \tau=\frac{T}{M}$ and a Doppler resolution of $\Delta \nu=\frac{1}{NT}$, which matches the DD resolution of the ODDM signal discussed in Section~\ref{sec:oddm}. Substituting these parameters into \eqref{eq:cdot_m}, we obtain the following corollary.
\begin{corollary}\label{cor:zca}
    The discrete chirp sequence $\boldsymbol{c}\in\mathbb{C}^{M\times1}$ with its $m$-th element given by
    \begin{align}\label{eq:c_m}
        c[m] = e^{j\pi\frac{m^2}{M}}, m\in\{0,\dots,M-1\},
    \end{align}
    has the ZCA property in \eqref{eq:zca} for all even $M$.
\end{corollary}

We then use $\boldsymbol{c}$ to construct the SRN-filtered chirp
\begin{align}\label{eq:ca}
    c_{a}(t)=\sum_{m=0}^{M-1}c[m]a\left(t-m\frac{T}{M}\right).
\end{align}
Through the usage of SRN pulse $a(t)$, $c_a(t)$ is naturally compatible with baseband modulation schemes, including ODDM. More details on the integration with ODDM will be discussed in Section~\ref{sec:dd_rcc_fmcw}. On the other hand, according to the ZCA property in Proposition \ref{prop:zca}, we can cyclically extend $c_a(t)$ to achieve a concise form of the delay ambiguity function. Thus, we construct the cyclically extended SRN-filtered chirp as
\begin{align}
    c_{a,ce}(t) = \sum_{n=-1}^{1}c_{a}(t-nT),
\end{align}
where a CP $c_{a}(t+T)$ and a CS $c_{a}(t-T)$ are attached. Then, the crosscorrelation between $c_{a,ce}(t)$ and $c_a(t)$ is computed as
\begin{align}
    A_{c_{a,ce},c_a}(\tau)&\triangleq\int_{-\infty}^\infty c_{a,ce}(t) c_a^*(t-\tau)dt
    \nonumber\\
    &=M\int_{-\infty}^\infty a(t) a^*(t-\tau)dt=Mg(\tau),
    \label{eq:A_ca_approx}
\end{align}
for delay shifts $\tau\in(-T,T)$. That is, $c_a(t)$ has a delay ambiguity function in the shape of the Nyquist pulse $g(\tau)$.

An example of $A_{c_a}(\tau)$ (green curve) is plotted in Fig.~\ref{fig:chirp_compression}, where we use a square-root-raised-cosine (SRRC) pulse for $a(t)$. Compared to $A_{c}(\tau)$ of a practical linear chirp, the sidelobes of $A_{c_a}(\tau)$ are significantly lower through the usage of the SRRC pulse. In addition, $g(\tau)$ naturally satisfies the Nyquist ISI criterion, facilitating ISAC waveform design.

\subsection{SRN-Filtered FMCW}\label{sec:rrc_fmcw}

To further achieve Doppler resolution, we use $c_{a}(t)$ as the radar pulse and define the baseband SRN-FMCW signal as
\begin{align}\label{eq:s_c}
    s_{c_a}(t) = \sum_{n=0}^{N-1}c_{a}(t-nT).
\end{align}
Notably, $s_{c_a}(t)$ can be seen as a pulse-Doppler radar signal with a $100\%$ duty cycle. Except for the first ($n=0$) and last ($n=N-1$) SRN-filtered chirp pulses, the $(n-1)$-th and $(n+1)$-th pulses naturally work as the cyclic extension of the $n$-th pulse $c_{a}(t-nT)$. As a result, only a framewise CS and CP are required, yielding
\begin{align}\label{eq:s_cce}
    s_{c_a,ce}(t) = \sum_{n=-1}^{N}c_{a}(t-nT).
\end{align}
An example of the transmitted sequence of the SRN-FMCW signal $s_{c_a,ce}(t)$ including a CP and a CS is given in the middle of Fig.~\ref{fig:cpcs}, where $l_{\max}\geq\tau_{\max}\frac{M}{T}$ is the CP length. The green rectangles represent random ODDM data symbols, whereas the blue rectangles---shown in varying shades---represent the discrete chirp sequence with a deterministic phase sweep.

\begin{figure*}[ht]
    \centering
    \includegraphics[width=1.7\columnwidth,trim={0 0 0 0},clip]{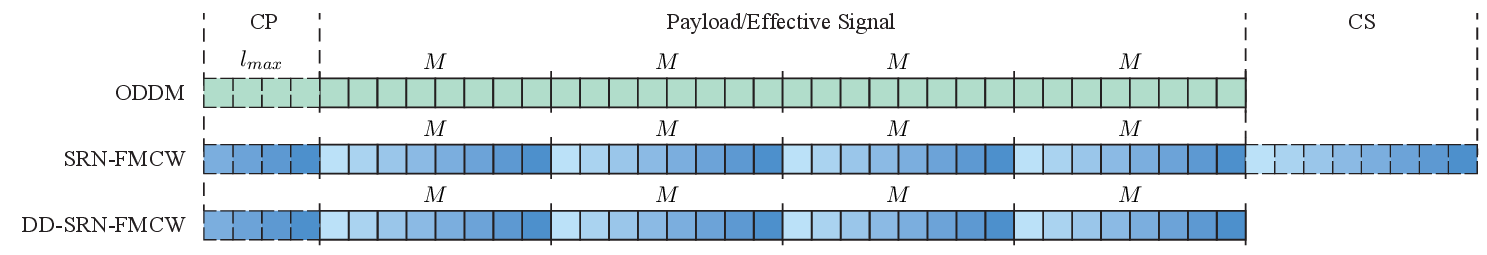}
    \caption{Time-domain transmitted sequences: ODDM, SRN-FMCW, and DD-SRN-FMCW. ($M=8,N=4$)}
    \label{fig:cpcs}
\end{figure*}

We now demonstrate how the DDR of a doubly-selective channel can be obtained using SRN-FMCW. Substituting \eqref{eq:s_c} to \eqref{eq:s2r_baseband}, the received SRN-FMCW signal can be expressed as
\begin{align}
    r_{c_a}(t) = \sum_{p=1}^{P} h_p  \sum_{\mathclap{n=-1}}^{N}c_{a}(t\!-\!\tau_p\!-\!nT) e^{j2\pi \nu_p(t-\tau_p)} + z(t).
\end{align}
Then, $c_a(t)$-based chirp compression is performed on $r_{c_a}(t)$. After critical sampling and reshaping, the \emph{radar data matrix} $\boldsymbol{R}_{c_a}\in\mathbb{C}^{M\times N}$ is obtained with $M$ fast-time bins and $N$ slow-time bins, whose $(m,n)$-th element is given by
\begin{align}\label{eq:R_ca}
    &R_{c_a}[m,n]=\sum_{p=0}^{P-1}h_p e^{j2\pi\frac{k_pn}{N}} \sum_{d=0}^{2Q} g\left(\left(d+\floor{l_p}-l_p\right)\frac{T}{M}\right) 
    \nonumber\\
    &\;\:\times\!\!\sum_{\tilde{m}=0}^{M-1} e^{j2\pi\frac{(\tilde{m}\!-\!l_p\!-\!d)k_p}{MN}} c\bigl[[\tilde{m}\!-\!m]_M\bigr]c\bigl[[\tilde{m}\!-\!\floor{l_p}\!-\!d]_{M}\bigr],
\end{align}
which is essentially the fast-slow time response of the channel. Finally, the channel DDR is computed by a row-wise DFT: $\boldsymbol{D}_{c_a}=\boldsymbol{R}_{c_a}\boldsymbol{F}_N$, with its $(\dot{m},n)$-th element given by
\begin{align}\label{eq:D_ca}
    &D_{c_a}[\dot{m},n] = \sum_{p=1}^{P} h_p\phi(k_p-n)\sum_{d=0}^{2Q} g\left((d+\floor{l_p}-l_p)\frac{T}{M}\right)
    \nonumber\\
    &\;\:\times\!\!\sum_{m=0}^{M-1}\!e^{j2\pi\frac{(m\!-\!l_p\!-\!d)k_p}{MN}} c\bigl[[\tilde{m}\!-\!\dot{m}]_M\bigr]c\bigl[[\tilde{m}\!-\!\floor{l_p}\!-\!d]_{M}\bigr].
\end{align}
By inspecting \eqref{eq:D_ca}, we can see that each path's response exhibits a Nyquist-shaped profile $g(\cdot)$ in delay and a Dirichlet-kernel-shaped profile $\phi(\cdot)$ in Doppler.

\subsection{DD-Domain Embedded SRN-Filtered FMCW}\label{sec:dd_rcc_fmcw}

This section presents DD-SRN-FMCW, an efficient DD-domain implementation of SRN-FMCW, to achieve compatibility with ODDM. We note that the SRN-FMCW signal $s_{c_a}(t)$ in \eqref{eq:s_c} is equivalent to the time-domain representation of an ODDM frame $\boldsymbol{X}_c\in\mathbb{C}^{M\times N}$, with its $(m,n)$-th symbol being
\begin{align}\label{eq:Xc=c}
    X_c[m,n] = 
    \begin{cases} 
        \sqrt{N E_c} c[m], & n=0,
        \\
        0, & \text{otherwise},
    \end{cases}
\end{align}
where we introduced $E_c=\mathbb{E}_t\left[s_{c_a}(t)\right]$ to represent the time-domain chirp power. That is, $s_{c_a}(t)$ in \eqref{eq:s_c} can be generated by embedding symbols in the DD domain as \eqref{eq:Xc=c}, followed by ODDM transmitter operations in Section~\ref{sec:oddm_tx}. We refer to $\boldsymbol{X}_c$ as the DD-SRN-FMCW frame.

To further integrate DD-SRN-FMCW with the ODDM receiver, we now introduce DD chirp compression. Let $\boldsymbol{Y}_c\in\mathbb{C}^{M\times N}$ denote the received DD-SRN-FMCW frame after the ODDM transmitter, channel, and receiver operations. Substituting \eqref{eq:Xc=c} into \eqref{eq:io_dd}, the $(m,n)$-th element of $\boldsymbol{Y}_c$ is
\begin{align}\label{eq:Y_c}
    Y_c[m,&n] = \sum_{p=1}^{P} h_p\phi(k_p-n)\sum_{d=0}^{2Q}e^{j2\pi\frac{(m-l_p-d)k_p}{MN}}
    \nonumber\\
    &\times g\left((d+\floor{l_p}-l_p)\frac{T}{M}\right) c\left[\left[m-\floor{l_p}-d\right]_{M}\right].
\end{align}
Then, chirp compression is performed in the DD domain as
\begin{align}\label{eq:ddr}
    D_c[\dot{m},n] \triangleq \sum_{\mathclap{m=0}}^{M-1} c\bigl[[m-\dot{m}]_M\bigr] Y_c[m,n].
\end{align}
Direct algebra shows that \eqref{eq:ddr} reduces to \eqref{eq:D_ca}, yielding the same DDR expression.

Compared to the direct $c_a(t)$-based chirp compression in Section~\ref{sec:rrc_fmcw}, the DD chirp compression by \eqref{eq:Y_c} and \eqref{eq:ddr} is decomposed into two cascaded stages: first matched filtering with the SRN pulse $a(t)$, then correlation with the discrete chirp sequence $\boldsymbol{c}$. As the first stage is part of ODDM receiver operations, DD chirp compression can be readily implemented using the ODDM receiver, followed by an $M$-point correlator.

Furthermore, as the correlation step is performed in the DD domain, we can exploit the DD quasi-periodicity: when the Doppler shift is absent ($k_p = 0$), samples from the $n=0$ Doppler column are periodic along the delay dimension, as seen from \eqref{eq:io_dd}. Since the cyclicity for $\boldsymbol{c}$ is guaranteed, we can \emph{waive the time-domain CS} in DD-SRN-FMCW to better align with ODDM. From the DD pilot design perspective, this also means the cyclic orthogonality of $\boldsymbol{c}$ can be exploited without a dedicated cyclic extension, making DD-SRN-FMCW different from other orthogonal-sequence-based pilots \cite{Shi2021MIMO-OTFS_DetPilot,Bondre2024DFRC-OTFS_RootMUSIC,Ma2025OTFS_LowPAPRPilots}. The transmitted sequences of ODDM, SRN-FMCW, and DD-SRN-FMCW are compared in Fig.~\ref{fig:cpcs}. DD-SRN-FMCW requires only a CP and aligns with the ODDM frame structure.

\section{Proposed ODDM-FMCW System for ISAC}\label{sec:oddm_fmcw}

In this section, we propose ODDM-FMCW for ISAC. We first introduce the ODDM-FMCW waveform by embedding DD-SRN-FMCW into an ODDM frame alongside data symbols. Then, we discuss the receiver signal processing schemes.

\subsection{ODDM-FMCW Waveform}\label{sec:oddm_fmcw_signal}

\begin{figure}[t]
    \centering
    \includegraphics[width=0.8\columnwidth,trim={-30 0 0 0},clip]{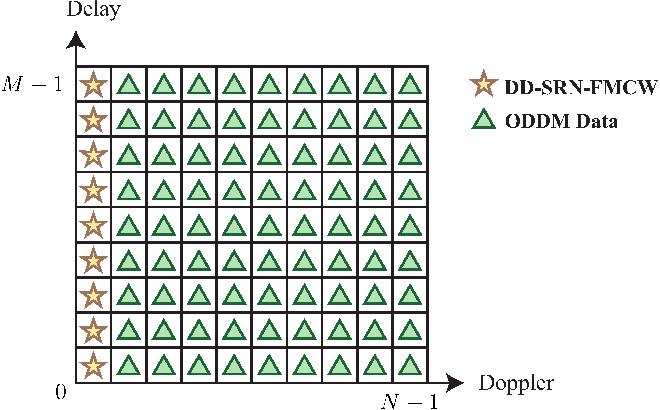}
    \caption{DD frame structure for ODDM-FMCW signals.}
    \label{fig:frame_oddmfmcw}
\end{figure}

To support ISAC operation in Fig.~\ref{fig:isac_chan}, the ODDM-FMCW frame $\boldsymbol{X}$ is constructed by superimposing a DD-SRN-FMCW frame $\boldsymbol{X}_c$ onto an ODDM data frame $\boldsymbol{X}_d\in\mathcal{A}^{M\times N}$, i.e.,
\begin{align}\label{eq:X}
    \boldsymbol{X}=\boldsymbol{X}_c+\boldsymbol{X}_d,
\end{align}
where each element of $\boldsymbol{X}_d$ is drawn from a constellation set $\mathcal{A}$ for data modulation. Fig.~\ref{fig:frame_oddmfmcw} illustrates the ODDM-FMCW frame structure. According to \eqref{eq:Xc=c}, the $n=0$ Doppler column is occupied by $\boldsymbol{X}_c$, used as the pilot for both channel estimation and sensing. The remaining columns are filled with data symbols $\boldsymbol{X}_d$, used for communication. No guard interval is allocated between DD-SRN-FMCW symbols and ODDM data symbols to avoid spectral efficiency loss. 

Given the linear transformation between DD domain and time domain, \eqref{eq:X} corresponds to a time-domain superimposed signal
\begin{align}\label{eq:s=scsd}
    s(t)=s_{c_a}(t)+s_d(t),
\end{align}
where $s_{c_a}(t)$ is the DD-SRN-FMCW signal and $s_d(t)$ is the ODDM data signal. We define the chirp-data-power-ratio (CDPR) as $\rho=E_c/E_s$, with chirp power $E_c=\mathbb{E}_t\left[s_{c_a}(t)\right]$ and data power $E_s=\mathbb{E}_t\left[s_d(t)\right]$. By tuning $\rho$, one can trade off sensing accuracy against communication BER.

The ODDM transmitter operations are performed on $\boldsymbol{X}$ to generate the baseband ODDM-FMCW signal. After propagation through the ISAC channel and ODDM receiver operations, the received frame $\boldsymbol{Y}$ is obtained by \eqref{eq:io_dd}. Given that no guard interval is employed, iterative interference cancellation is employed at both the collocated sensing receiver and the communication receiver. For the collocated sensing receiver, a DAS algorithm is used to enable super-resolution sensing. For the communication receiver, a JCEDD algorithm is used to perform data detection. In the following sections, we detail these algorithms one after the other.

\subsection{Data-Aided Sensing at Collocated Sensing Receiver}\label{sec:sensing}

DAS is performed at the collocated sensing receiver to estimate the DD parameters of the channel paths. We improve the OMP algorithm \cite{Rasheed2020OTFS_MU_OMP} with grid evolution for low-complexity and super-resolution sensing. The algorithm successively estimates the parameters $\left\{\hat{h}_p,\hat{l}_p,\hat{k}_p\right\}$ for path $p\in\{1,\dots,P_{\max}\}$ with $P_{\max}$ being a predefined maximum path number, where grid evolution is performed as an inner loop to refine the DD resolution. The signal components from the previously estimated paths are subtracted from the received signal to cancel their interference.

Let $\Delta\boldsymbol{Y}_{p}$ denote the path cancellation residual before the estimation of the $p$-th path, computed by
\begin{align}\label{eq:DeltaY}
    \Delta\boldsymbol{Y}_{p}\!=\boldsymbol{Y}-\sum_{p'=1}^{p-1}\!\hat{h}_{p'} \boldsymbol{a}\!\left(\hat{l}_{p'},\hat{k}_{p'}\right).
\end{align}
We have $\Delta\boldsymbol{Y}_{1}=\boldsymbol{Y}$ before estimating the first path. Due to the high processing gain of the DD-SRN-FMCW component, we can compute the residual DDR $\Delta\boldsymbol{D}_p$ by substituting $\Delta\boldsymbol{Y}_{p}$ to \eqref{eq:ddr}. The highest response in $\Delta\boldsymbol{D}_p$ gives the integer DD shifts $\left\{\hat{l}_p^{(0)},\hat{k}_p^{(0)}\right\}$. Then, we perform OMP on a virtual grid $\Lambda$ with evolving DD resolution. In particular, for the $\ell$-th grid evolution, a matching pursuit is performed as
\begin{align}
    \argmax_{\hat{l}_p^{(\ell)},\hat{k}_p^{(\ell)}\in\Lambda^{(\ell)}}\frac{\left|\boldsymbol{a}\!\left(\hat{l}_p^{(\ell)},\hat{k}_p^{(\ell)}\right)^\herm\!\mathrm{vec}\left(\Delta\boldsymbol{Y}_p\right)\right|^2}{\left\lVert\boldsymbol{a}\!\left(\hat{l}_p^{(\ell)},\hat{k}_p^{(\ell)}\right)\right\rVert^2},
    \label{eq:omp}
\end{align}where $\boldsymbol{a}\!\left(\hat{l}_p,\hat{k}_p\right)\in\mathbb{C}^{MN\times1}$ denotes the atom, i.e., the noiseless response produced by a unit-gain path at delay $\hat{l}_p$ and Doppler $\hat{k}_p$, computed by \eqref{eq:Y_c}. The virtual grid $\Lambda^{(\ell)}$ for the $\ell$-th evolution is defined to be centered at the last-evolution estimates $\left\{\hat{l}_p^{(\ell-1)},\hat{k}_p^{(\ell-1)}\right\}$ with a halved DD resolution. After evolving the virtual grid $\mathcal{L}$ times, the final estimation $\left\{\hat{l}_p,\hat{k}_p\right\}=\left\{\hat{l}_p^{(\mathcal{L})},\hat{k}_p^{(\mathcal{L})}\right\}$ lies on a refined DD grid $\Lambda^{(\mathcal{L})}$ with spacings reduced by a factor of $2^{-\mathcal{L}}$.

After refining the DD estimates of the $p$-th path by grid evolution, the coefficients of all previously estimated paths are updated by least squares\cite{Rasheed2020OTFS_MU_OMP}:
\begin{align}
    &\left[\hat{h}_1,\dots,\hat{h}_p\right]^\tran 
    \nonumber\\
    &\ \ \qquad{=} \left[\boldsymbol{a}\!\left(\hat{l}_1^{(\mathcal{L})},\hat{k}_1^{(\mathcal{L})}\right),\dots,\boldsymbol{a}\!\left(\hat{l}_p^{(\mathcal{L})},\hat{k}_p^{(\mathcal{L})}\right)\right]^\dagger\!\mathrm{vec}\left(\boldsymbol{Y}\right).
    \label{eq:omp_ls}
\end{align}
The successive path estimation continues until $P_{\max}$ is reached or $\norm{\Delta\boldsymbol{Y}_{p}}^2$ stops decreasing, and we denote the resulting number of estimated paths by $\hat{P}\leq P_{\max}$. We summarize the OMP algorithm with grid evolution in Algorithm~\ref{alg:omp_ge}.

\begin{algorithm}[t]
  \caption{OMP with Grid Evolution}
  \label{alg:omp_ge}
  \begin{algorithmic}[1]
    \Input $\boldsymbol{Y},\boldsymbol{c},P_{\max}$.
    \Initialize $p\gets0$.
    \Repeat
        \State $p \gets p+1$.
        \State Obatin $\Delta\boldsymbol{Y}_{p}$ by \eqref{eq:DeltaY}.
        \State Obatin $\Delta\boldsymbol{D}_p$ by substituting $\Delta\boldsymbol{Y}_{p}$ to \eqref{eq:ddr}.
        \State Obatin $\left\{\hat{l}_p^{(0)},\hat{k}_p^{(0)}\right\}$ by the highest response in $\Delta\boldsymbol{D}_p$.
        \For{$\ell=1,\ldots,\mathcal{L}$}
            \State Construct $\Lambda^{\ell}$ based on $\left\{\hat{l}_p^{(\ell-1)},\hat{k}_p^{(\ell-1)}\right\}$.
            \State Obatin $\left\{\hat{l}_p^{(\ell)},\hat{k}_p^{(\ell)}\right\}$ by \eqref{eq:omp}.
        \EndFor
        \State $\hat{l}_p\gets \hat{l}_p^{(\mathcal{L})}$, $\hat{k}_p\gets \hat{k}_p^{(\mathcal{L})}$.
        \State Update $\left\{\hat{h}_1,\dots,\hat{h}_p\right\}$ by \eqref{eq:omp_ls}.
    \Until{$p = P_{\max}$ or $\norm{\Delta\boldsymbol{Y}_{p+1}}^2 \ge \norm{\Delta\boldsymbol{Y}_{p}}^2$.}
    \State $\hat{P} \gets p$.
    \Output $\left\{\hat{h}_p,\hat{l}_p,\hat{k}_p\right\}$ for $p\in\{1,\dots,\hat{P}\}$.
  \end{algorithmic}
\end{algorithm}

\begin{algorithm}[t]
  \caption{Data-Aided Sensing}
  \label{alg:das}
  \begin{algorithmic}[1]
    \Input $\boldsymbol{Y},\boldsymbol{c},\boldsymbol{X}_d,P_{\max}$.
    \Initialize $\hat{\boldsymbol{Y}}_c \gets \boldsymbol{Y}$.
    \For{$i=1,\ldots,I_\text{DAS}$}
        \State $i \gets i+1$.
        \State Estimate the channel by applying Algorithm~\ref{alg:omp_ge} to $\hat{\boldsymbol{Y}}_c$.
        \State Estimate $\hat{\boldsymbol{Y}}_d$ by substituting $\boldsymbol{X}_d$ to \eqref{eq:io_dd}.
        \State $\hat{\boldsymbol{Y}}_c \gets \boldsymbol{Y} - \hat{\boldsymbol{Y}}_d$.
    \EndFor
    \Output $\left\{\hat{h}_p,\hat{l}_p,\hat{k}_p\right\}$ for $p\in\{1,\dots,\hat{P}\}$.
  \end{algorithmic}
\end{algorithm}

The OMP algorithm described above treats superimposed data symbols as noise. However, since the data symbols are known to the collocated sensing receiver, the channel distorted component of $\boldsymbol{X}_d$ can be computed by \eqref{eq:Y_c} using the estimated channel parameters and subtracted from $\boldsymbol{Y}$ to reduce interference from data. After that, OMP with path cancellation is employed again to achieve better DD estimates. This process is iterated to progressively refine the estimation, as summarized in Algorithm~\ref{alg:das}.

Here we derive the complexity order of the DAS algorithm, which operates by iteratively cancelling the data component and refining the channel parameters via an OMP-based subroutine. Within each DAS iteration, the OMP-based subroutine estimates the channel parameters path by path. For each path, coarse delay and Doppler estimates are first obtained by locating the maximum of the residual DDR $\Delta\boldsymbol{D}_p$ in \eqref{eq:ddr}. The circular correlation in \eqref{eq:ddr} has an efficient FFT-based implementation, which has a complexity of $\mathcal{O}(NM\log_2M)$. Then, matching pursuit as in \eqref{eq:omp} is performed to refine the estimates with the grid evolving $\mathcal{L}$ times, incurring a complexity of $\mathcal{O}(\mathcal{L}MN)$. After estimating the delay and Doppler of the $p$-th path, the coefficients of the first $p$ paths are updated by \eqref{eq:omp_ls} with complexity $\mathcal{O}(p^3+MNp^2)$, which sums to $\mathcal{O}(\hat{P}^4+MN\hat{P}^3)$ over all $\hat{P}$ paths. Therefore, the total complexity of the proposed DAS algorithm is $\mathcal{O}\bigl(I_\mathrm{DAS}\bigl(\hat{P}NM\bigl(\log_2M+\mathcal{L}+\hat{P}^2\bigr)+\hat{P}^4\bigr)\bigr)$. Notably, $\hat{P}$ is typically small in sparse doubly-selective channels. The DDR-based coarse estimation significantly shrinks the candidate atom dictionary, while the subsequent grid evolution dynamically refines this compact dictionary to achieve high estimation accuracy without incurring excessive computational complexity.

\subsection{Joint Channel Estimation and Data Detection at Communication Receiver}\label{sec:jcedd}

Since both $\boldsymbol{X}_d$ and CSI are unknown to the communication receiver, we perform JCEDD by combining the modified OMP algorithm with the soft SIC-MMSE detector \cite{Li2024SICMMSE_Turbo}, which offers superior BER compared to other common detectors like message passing and iterative maximum ratio combining \cite{Huang2024ODDM_Performance}.

We first brief the soft SIC-MMSE detector assuming perfect CSI, which performs low-complexity MMSE equalization using sub-channel matrices. In the original SIC-MMSE algorithm \cite{Li2024SICMMSE_Turbo}, each time-domain sample $s[q]$ is associated with a fixed-length spreading vector $\boldsymbol{g}_q'\in\mathbb{C}^{(l_{\max}+2Q+1)}$ to capture its channel response. Then, a sub-channel matrix $\boldsymbol{G}_q'\in\mathbb{C}^{(l_{\max}+2Q+1)\times(2l_{\max}+4Q+3)}$ is constructed around $\boldsymbol{g}_q'$ to cover the channel response relevant to $s[q]$ and the interfering samples. The relation between $\boldsymbol{g}_q'$, $\boldsymbol{G}_q'$, and the full channel matrix $\boldsymbol{G}$ is illustrated by the top two graphs in Fig.~\ref{fig:subio}.

\begin{figure}[t]
    \centering
    \includegraphics[width=0.8\columnwidth,trim={0 0 0 0},clip]{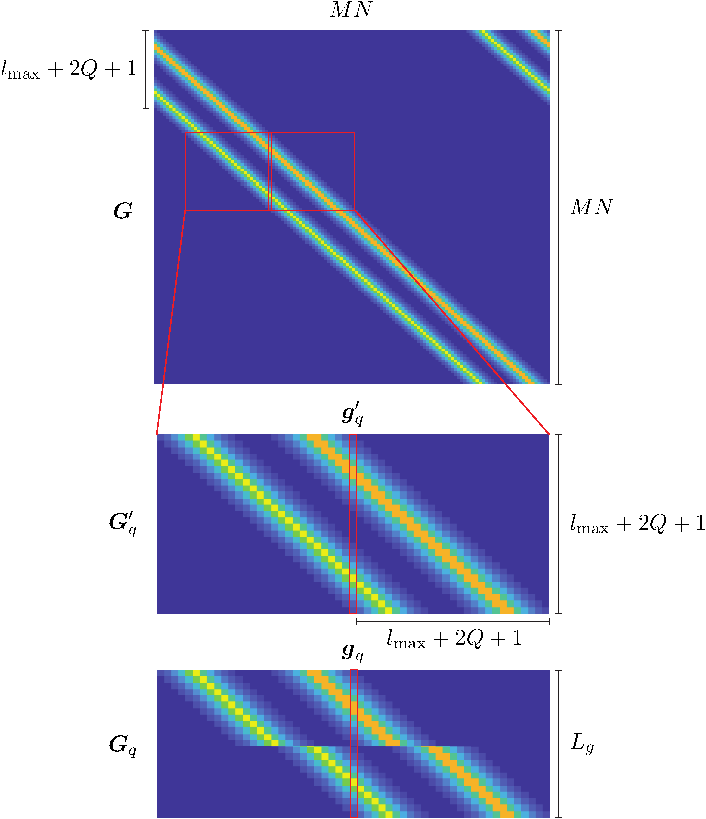}
    \caption{Construction of a sub-channel matrix for SIC-MMSE.}
    \label{fig:subio}
\end{figure}

In this paper, we use a modified sub-channel matrix construction method to adaptively reduce the size of $\boldsymbol{G}_q'$. In particular, we collapse the zero elements in $\boldsymbol{g}_q'$ to obtain a dense spreading vector $\boldsymbol{g}_q\in\mathbb{C}^{L_g}$, where $L_g\le\min\{l_{\max}+2Q+1,P(2Q+1)\}$ is a constant representing the number of non-zero elements in $\boldsymbol{g}_q'$. The corresponding sub-channel matrix $\boldsymbol{G}_q\in\mathbb{C}^{L_g\times(2l_{\max}+4Q+3)}$ is constructed with height $L_g$. This modification is illustrated by the bottom graph in Fig.~\ref{fig:subio}. It can be inferred that the size reduction to $L_g$ is pronounced for sparse channels with large delay spreads, which is precisely the regime considered in our ISAC model.

Based on $\boldsymbol{G}_q$, we can define a sub-input-output relation for each $s[q]$ to \emph{only} include the channel distorted components of $s[q]$ itself and the interfering samples, allowing MMSE equalization with small sub-channel matrices. In practical implementation, interference cancellation is used to efficiently express the sub-input-output relation \cite{Li2024SICMMSE_Turbo}. Let $\Delta\boldsymbol{r}=\boldsymbol{r}-\boldsymbol{G}\hat{\boldsymbol{s}}$ denote the residual interference plus noise vector, where $\hat{\boldsymbol{s}}$ is the a priori mean of the transmitted sequence $\boldsymbol{s}$. For each sample $\hat{s}[q]$, we express the estimated received vector as $\tilde{\boldsymbol{r}}_q = \Delta\boldsymbol{r}_q + \boldsymbol{g}_q \hat{s}[q]$, with $\Delta\boldsymbol{r}_q$ being the portion of $\Delta\boldsymbol{r}$ corresponding to $\boldsymbol{g}_q$. Then, sample-by-sample soft MMSE equalization can be performed as $\tilde{s}[q] = \boldsymbol{g}_q^\herm\left(\boldsymbol{G}_q\boldsymbol{V}_q\boldsymbol{G}_q^\herm+\sigma_z^2\right)^{-1} \tilde{\boldsymbol{r}}_q$, where $\boldsymbol{V}_q$ is the corresponding a priori error covariance.

After MMSE equalization, cross-domain symbol-wise maximum likelihood (ML) detection is performed on $\tilde{s}[q]$ to get the a posteriori estimate. The soft SIC-MMSE detector is run iteratively, using a posteriori estimates from the last iteration as a priori inputs for the next, thereby improving detection accuracy. The overall SIC-MMSE algorithm is summarized in Algorithm \ref{alg:sicmmse}, where $\eta_\mathcal{A}(\cdot)$ denotes the soft-decision ML detector over constellation set $\mathcal{A}$.

\begin{algorithm}[t]
  \caption{Soft SIC-MMSE with Perfect CSI}
  \label{alg:sicmmse}
  \begin{algorithmic}[1]
    \Input $\boldsymbol{Y},\boldsymbol{G},\sigma_z^2,\mathcal{A}$.
    \Initialize $\hat{\boldsymbol{X}}\gets\boldsymbol{0}^{M\times N},\hat{\boldsymbol{s}}\gets\boldsymbol{0}^{MN}$.
    \For{$i=1,\ldots,I_\mathrm{DET}$}
        \State $i \gets i+1$.
        \For{$m=0,\ldots,M-1$}
            \State $\Delta\boldsymbol{r}\gets \boldsymbol{r}-\boldsymbol{G}\hat{\boldsymbol{s}}$.
            \State $\tilde{\boldsymbol{s}}_m\gets\boldsymbol{0}^{N}$.
            \For{$\dot{n}=0,\ldots,N-1$}
                \State $q\gets \dot{n}M+m$.
                \State Extract $\boldsymbol{g}_q$ and $\boldsymbol{G}_q$ from $\boldsymbol{G}$.
                \State $\tilde{\boldsymbol{r}}_q \gets \Delta\boldsymbol{r}_q + \boldsymbol{g}_q \hat{s}[q]$. 
                \State $\tilde{s}_m[\dot{n}] \gets \boldsymbol{g}_q^\herm\left(\boldsymbol{G}_q\boldsymbol{V}_q\boldsymbol{G}_q^\herm+\sigma_z^2\right)^{-1} \tilde{\boldsymbol{r}}_q$.
            \EndFor
            \State $\tilde{\boldsymbol{x}}_m\gets\boldsymbol{F}_N \tilde{\boldsymbol{s}}_m$.
            \For{$n=0,\ldots,N-1$}
                \State $\hat{X}[m,n]\gets\eta_\mathcal{A}(\tilde{x}_m[n])$.
            \EndFor
            \State $\left\{\hat{s}[\dot{n}M+m]\right\}_{\dot{n}=0}^{N-1}\gets\boldsymbol{F}_N^\herm \hat{\boldsymbol{X}}[m,:]$.
        \EndFor
    \EndFor
    \Output $\hat{\boldsymbol{X}}$.
  \end{algorithmic}
\end{algorithm}

When the channel parameters are unknown, we integrate Algorithm \ref{alg:omp_ge} for channel estimation with Algorithm \ref{alg:sicmmse} for data detection to perform JCEDD for ODDM-FMCW. The resulting JCEDD algorithm is summarized in Algorithm \ref{alg:jcedd}. Given the a priori estimated frame $\hat{\boldsymbol{X}}$, we first separate the received pilot component $\hat{\boldsymbol{Y}}_c$ by treating the channel distorted components of the detected data as known interference. The channel parameters are then estimated by applying Algorithm \ref{alg:omp_ge} to $\hat{\boldsymbol{Y}}_c$. Based on the estimated channel, we construct the effective channel matrix $\boldsymbol{G}$ and perform data detection using the soft SIC-MMSE detector, improving $\hat{\boldsymbol{X}}$. This process is iterated to progressively refine both channel estimates and data estimates. Notably, the SIC-MMSE module in JCEDD differs from Algorithm \ref{alg:sicmmse} in that the symbol-wise ML update of $\hat{\boldsymbol{X}}$ excludes the known DD-SRN-FMCW symbols at $n=0$, as indicated in line~20 of Algorithm~\ref{alg:jcedd}.

\begin{algorithm}[t]
  \caption{Joint Channel Estimation and Data Detection}
  \label{alg:jcedd}
  \begin{algorithmic}[1]
    \Input $\boldsymbol{Y},\boldsymbol{X}_c,\sigma_z^2,\mathcal{A}$.
    \Initialize $\hat{\boldsymbol{X}}_d\gets\boldsymbol{0}^{M\times N},\hat{\boldsymbol{X}}\gets\boldsymbol{X}_c,\hat{\boldsymbol{s}}\gets\mathrm{vec}\left(\hat{\boldsymbol{X}}\boldsymbol{F}_N^\herm\right)$.
    \For{$i=1,\ldots,I_\text{JCEDD}$}
        \State $i \gets i+1$.
        \State Estimate $\hat{\boldsymbol{Y}}_d$ by substituting $\hat{\boldsymbol{X}}_d$ to \eqref{eq:io_dd}.
        \State $\hat{\boldsymbol{Y}}_c\gets\boldsymbol{Y}-\hat{\boldsymbol{Y}}_d$.
        \State Estimate the channel by applying Algorithm~\ref{alg:omp_ge} to $\hat{\boldsymbol{Y}}_c$.
        \State Construct $\boldsymbol{G}$ based on the estimated channel.
        \For{$m=0,\ldots,M-1$}
            \State $\Delta\boldsymbol{r}\gets \boldsymbol{r}-\boldsymbol{G}\hat{\boldsymbol{s}}$.
            \State $\tilde{\boldsymbol{s}}_m\gets\boldsymbol{0}^{N}$.
            \For{$\dot{n}=0,\ldots,N-1$}
                \State $q\gets \dot{n}M+m$.
                \State Extract $\boldsymbol{g}_q$ and $\boldsymbol{G}_q$ from $\boldsymbol{G}$.
                \State $\tilde{\boldsymbol{r}}_q \gets \Delta\boldsymbol{r}_q + \boldsymbol{g}_q \hat{s}[q]$. 
                \State $\tilde{s}_m[\dot{n}] \gets \boldsymbol{g}_q^\herm\left(\boldsymbol{G}_q\boldsymbol{V}_q\boldsymbol{G}_q^\herm+\sigma_z^2\right)^{-1} \tilde{\boldsymbol{r}}_q$.
            \EndFor
            \State $\tilde{\boldsymbol{x}}_m\gets\boldsymbol{F}_N \tilde{\boldsymbol{s}}_m$.
            \For{$n=1,\ldots,N-1$}
                \State $\hat{X}[m,n]\gets\eta_\mathcal{A}(\tilde{x}_m[n])$.
            \EndFor
            \State $\left\{\hat{s}[\dot{n}M+m]\right\}_{\dot{n}=0}^{N-1}\gets\boldsymbol{F}_N^\herm \hat{\boldsymbol{X}}[m,:]$.
        \EndFor
        \State $\hat{\boldsymbol{X}}_d \gets \hat{\boldsymbol{X}} - \boldsymbol{X}_c$.
    \EndFor
    \Output $\hat{\boldsymbol{X}}_d$.
  \end{algorithmic}
\end{algorithm}

We now analyze the complexity order of the JCEDD algorithm in Algorithm~\ref{alg:jcedd}, which can be viewed as a soft SIC-MMSE detector (Algorithm~\ref{alg:sicmmse}) where each iteration invokes the OMP-based channel estimator in Algorithm~\ref{alg:omp_ge}. The complexity order of Algorithm~\ref{alg:omp_ge} was derived in Section \ref{sec:sensing}. Meanwhile, the complexity of a soft SIC-MMSE detector is $\mathcal{O}\left(I_\mathrm{DET}MN\left(\log_2N+L_g^3\right)\right)$ \cite{Huang2024ODDM_Performance,Li2024SICMMSE_Turbo}. Therefore, the total complexity of the proposed JCEDD algorithm is $\mathcal{O}\bigl(I_\mathrm{JCEDD}\bigl(MN\bigl(\log_2N{+}L_g^3{+}\hat{P}\log_2M{+}\mathcal{L}{+}\hat{P}^2\bigr){+}\hat{P}^4\bigr)\bigr)$.

Notably, the proposed ISAC receiver architecture adopts OMP with grid evolution for sensing/channel estimation and soft SIC-MMSE for data detection because this combination offers a favorable performance-complexity tradeoff for doubly selective channels with large delay spreads. OMP with grid evolution achieves super-resolution sensing with complexity that grows linearly in the DD frame size $MN$ and only polynomially in the typically small number of dominant paths $\hat{P}$, thereby avoiding the cubic-in-$MN$ cost of sparse Bayesian learning \cite{Wei2022SBL,Shan2025ODDM_ChanEst_GridRefinement} and subspace-based \cite{Bondre2024DFRC-OTFS_RootMUSIC} estimators. On the communication side, the soft SIC-MMSE detector has been shown to outperform message passing and iterative maximum ratio combining \cite{Huang2024ODDM_Performance}, while remaining far less complex than orthogonal approximate message passing \cite{Ma2017OAMP}, whose complexity scales cubically in $MN$. These properties make the OMP/SIC-MMSE pair a natural choice as the core building blocks of the proposed DAS and JCEDD frameworks.

\section{Performance Analysis}\label{sec:performance_analysis}

In this section, we evaluate the performance of the proposed ODDM-FMCW system. We first examine its PAPR and spectral characteristics. We then derive the ambiguity function and CRBs for delay and Doppler estimation to quantify the sensing performance of DD-SRN-FMCW.

\subsection{Peak-to-Average Power Ratio (PAPR)}
We first aim to analyze the PAPR of the ODDM-FMCW signal $s(t)$, defined as $\gamma = \frac{\max\abs{s(t)}^2}{\mathbb{E}_{t}\left[\abs{s(t)}^2\right]}$. Given \eqref{eq:s=scsd}, the average power of $s(t)$ is $\mathbb{E}_{t}\left[\abs{s(t)}^2\right] = E_c+E_s$.

Since $\gamma$ depends heavily on the specific realization of the data component $s_d(t)$, the statistical distribution of $\gamma$ is of interest, often characterized by the complementary cumulative distribution function (CCDF) $p(\gamma>\gamma_0)$ with $\gamma_0$ being some threshold \cite{Surabhi2019PAPR_OTFS,Rahmatallah2013PAPR_Reduction_OFDM_Survey,Jiang2008PAPR_OFDM_EVT}. As $N$ increases, the central limit theorem implies that the samples of $s_d(t)$ converge to a pointwise complex Gaussian distribution \cite{Surabhi2019PAPR_OTFS,Rahmatallah2013PAPR_Reduction_OFDM_Survey,Jiang2008PAPR_OFDM_EVT}. On the other hand, the deterministic DD-SRN-FMCW signal $s_{c_a}(t)$, generated from the constant-amplitude $\boldsymbol{c}$, retains a near-constant envelope. Therefore, we assume $s(t){\sim}\mathcal{CN}\left(s_{c_a}(t),E_s\right)$ so that $\abs{s(t)}$ is Rician distributed \cite{Tse2005Fundamentals_WC,Richards2022Fundamentals_RadarSignalProcessing}. Equivalently, $\abs{s(t)}^2/\mathbb{E}_{t}\left[\abs{s(t)}^2\right]$ is a non-central $\chi^2$ process with two degrees of freedom, whose tail probability is a function of the CDPR $\rho$:
\begin{align}\label{eq:tail_oddmfmcw}
    p\left(\frac{\abs{s(t)}^2}{\mathbb{E}_{t}\left[\abs{s(t)}^2\right]}{>}\gamma_0\right)=Q_1\left(\sqrt{2\rho},\sqrt{2\gamma_0\left(1\!+\!\rho\right)}\right),
\end{align}
with $Q_1(\cdot,\cdot)$ being the Marcum Q-function of order one. Finally, $p(\gamma\!>\!\gamma_0)$ can be approximated by considering $MN$ i.i.d. samples of $s(t)$:\footnote{
    As $s(t)$ is bandlimited, its samples are in fact correlated, so the finite-dimensional distribution of $\abs{s(t)}^2/\mathbb{E}_{t}\left[\abs{s(t)}^2\right]$ does \emph{not} converge to that of a jointly $\chi^2$ process. More accurate approximation is non-trivial \cite{Jiang2008PAPR_OFDM_EVT} and lies beyond the scope of this paper.
}
\begin{align}\label{eq:ccdf_oddmfmcw}
    p(\gamma\!>\!\gamma_0)\approx 1\!-\!\left(\!1\!-\!Q_1\!\left(\!\sqrt{2\rho},\sqrt{2\gamma_0\left(1\!+\!\rho\right)}\right)\!\!\right)^{\!\!MN}.
\end{align}

The analytical CCDF in \eqref{eq:ccdf_oddmfmcw} is plotted in Fig.~\ref{fig:ccdf_analysis} using the parameters in Section~\ref{sec:numerical}, where the simulated CCDF is also shown for comparison. It can be observed that the analytical CCDF matches the simulation results well. In addition, with an increased CDPR $\rho$, ODDM-FMCW can have a lower PAPR cutoff value. This is because $p(\gamma\!>\!\gamma_0)$ in \eqref{eq:ccdf_oddmfmcw} is a monotonically decreasing function of $\rho$ provided that $Q_1(\cdot,\cdot)$ decreases with its second argument and increases only with the square-root of its first argument \cite{Marcum1960StatTheory}. \emph{Therefore, increasing the DD-SRN-FMCW component improves the PAPR of ODDM-FMCW.} As will be further shown in Section~\ref{sec:numerical}, the PAPR of ODDM-FMCW is lower than that of pure ODDM data signal and significantly lower than that of ODDM-DDIP.

\begin{figure}[t]
    \centering
    \includegraphics[width=0.96\columnwidth,trim={0 0 0 0},clip]{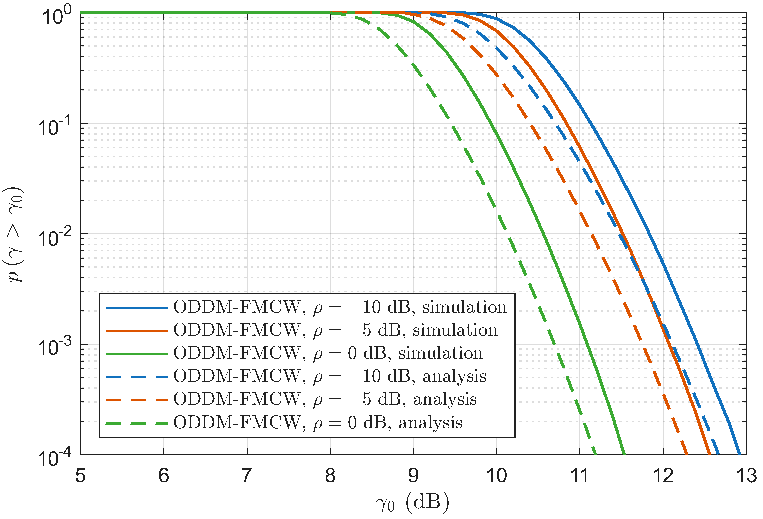}
    \caption{CCDF of ODDM-FMCW signals.}
    \label{fig:ccdf_analysis}
\end{figure}

\subsection{Spectrum}\label{sec:spectrum}

Now we derive the PSD of the ODDM-FMCW signal $s(t)$. The frequency response of $s(t)$ can be written as
\begin{align}\label{eq:S_f}
    S(f)=\sum_{m=0}^{M-1}\sum_{n=0}^{N-1} X[m,n] B_{m,n}(f),
\end{align}
where $B_{m,n}(f) = \sqrt{N} A(f) e^{-j2\pi f\frac{T}{M}m} \phi(n-NTf)$ is the frequency response of the ODDM basis function \cite{Shafie2025ODDM_Spectrum} and $A(f)$ is the frequency response of the SRN pulse $a(t)$ \cite{Shafie2024DDOP_TF_Loc}.

Substituting $\boldsymbol{X}_c$ to \eqref{eq:S_f}, the PSD of the DD-SRN-FMCW component $s_{c_a}(t)$ is derived as
\begin{align}\label{eq:psd_sca}
    \abs{S_{c_a}(f)}^2 = E_c N \varpi(f) \abs{A(f)}^2 \abs{\phi(-NTf)}^2,
\end{align}
where $\varpi(f)=\sum_{m_1=0}^{M-1}\sum_{m_2=0}^{M-1}c[m_1]c^*[m_2] e^{-j2\pi f\frac{(m_1-m_2)T}{M}}$. Similarly, the average PSD of the data component $s_d(t)$ is
\begin{align}\label{eq:psd_sd}
    \mathbb{E}[\abs{S_d(f)}^2] = E_s NM\abs{A(f)}^2\sum_{n=1}^{N-1} \abs{\phi(n-NTf)}^2.
\end{align}
Finally, the average PSD of ODDM-FMCW is given by
\begin{align}
    \mathbb{E}[\abs{S(f)}^2]=\abs{S_{c_a}(f)}^2+\mathbb{E}[\abs{S_d(f)}^2].
\end{align}
By \eqref{eq:psd_sca} and \eqref{eq:psd_sd}, it can be oberved that $S_{c_a}(f)$ is modulated onto the $n=0$ Doppler tone $\phi(-NTf)$, while $S_d(f)$ is modulated onto the other Doppler tones $\phi(n-NTf)$ for $n=1,\dots,N-1$. Both $\abs{S_{c_a}(f)}^2$ and $\mathbb{E}[\abs{S_d(f)}^2]$ have an envelope $\abs{A(f)}^2$, which can be leveraged for OOBE control. 

In Fig.~\ref{fig:psd_fmcw}, we compare the PSD of DD-SRN-FMCW with that of the linear FMCW signal $s_c(t)$ in \eqref{eq:sc}, denoted by $\abs{S_c(f)}^2$. A SRRC pulse is used as $a(t)$. The PSDs for SRN-filtered chirp and linear chirp are also plotted, denoted by $\abs{C_a(f)}^2$ and $\abs{C(f)}^2$, respectively. All the signals use a nominal bandwidth of $\frac{M}{T}=\SI{480}{\kilo\hertz}$, based on $M=32$ and $T=\SI{66.67}{\micro\second}$. Linear chirp and the linear FMCW signal consider the chirp parameters used by Corollary \ref{cor:zca}. Both $\abs{S_c(f)}^2$ and $\abs{S_{c_a}(f)}^2$ use a pulse train of length $N=4$. With the repetition of pulses, both $\abs{S_c(f)}^2$ and $\abs{S_{c_a}(f)}^2$ are modulated onto the $n=0$ Doppler tone $\phi(-NTf)$. Due to the slowly attenuating chirp spectrum $\abs{C(f)}^2$, the linear FMCW spectrum $\abs{S_c(f)}^2$ shows a high OOBE. On the other hand, through the usage of the SRRC pulse, the $\abs{A(f)}^2$ envelope of $\abs{C_a(f)}^2$ and $\abs{S_{c_a}(f)}^2$ provides significantly reduced OOBE.

\begin{figure}[t]
    \centering
    \includegraphics[width=0.95\columnwidth,trim={0 0 0 0},clip]{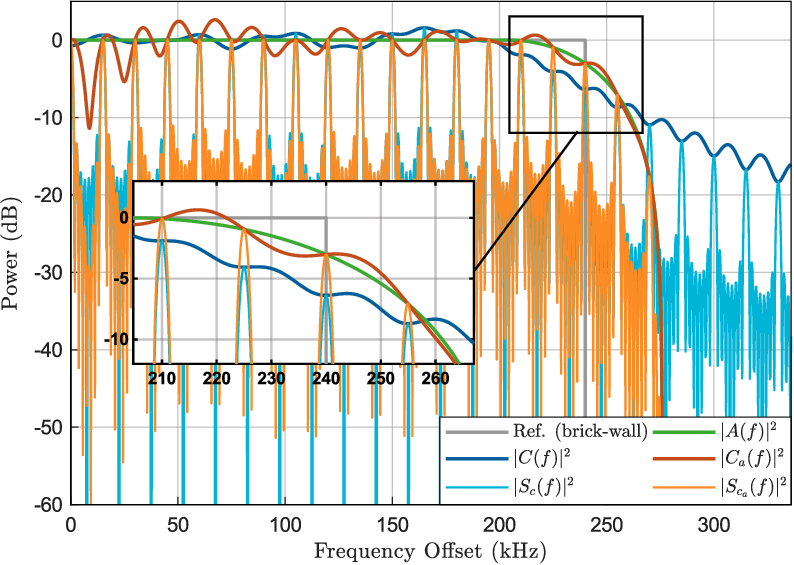}
    \caption{PSD of different sensing signals ($M=32,N=8$).}
    \label{fig:psd_fmcw}
\end{figure}

\subsection{Ambiguity Function}\label{sec:af_analysis}
To evaluate the sensing capability of the DD-SRN-FMCW signal $s_{c_a}(t)$, we present its ambiguity function. Following the discussion in Section~\ref{sec:linearfmcw2rrcfmcw}, we consider the cyclic extended signal $s_{c_a,ce}(t)$ in \eqref{eq:s_cce} as the transmitted signal. Then, the cross-ambiguity function between $s_{c_a,ce}(t)$ and $s_{c_a}(t)$ is
\begin{align}\label{eq:af_sce_sc_1}
    A(\tau,&\nu) \triangleq \int_{-\infty}^\infty s_{c_a,ce}(t)s_{c_a}^*(t-\tau)e^{-j2\pi\nu(t-\tau)}dt
    \nonumber\\
    &\!\!\!\!\!= N\phi(-\nu NT)\sum_{\mathclap{\zeta=0}}^{M-1}\,\sum_{m=0}^{M-1} e^{-j2\pi\nu\frac{T}{M}m} c[m] c^*\bigl[[m+\zeta]_{M}\bigr] 
    \nonumber\\
    &\!\!\!\times \int_{-\infty}^\infty a\left(t\right)  a^*\left(t-\zeta \frac{T}{M}-\tau\right)e^{j2\pi\nu (\tau-t)} dt,
\end{align}
where the proof is omitted for brevity. Fig.~\ref{fig:af_all} shows an example of $A(\tau,\nu)$ with a SRRC pulse used for $a(t)$. A sharp mainlobe and well-controlled sidelobes are observed, confirming the sensing capability of $s_{c_a}(t)$.

Since the SRN pulse $a(t)$ is well-localized in time, we can make the linear approximation $\int a(t)a^*(t-\tau) e^{-j2\pi\nu t}dt\approx g(\tau)$ following \cite{Lin2022OrthogonalModulation,Tong2024ODDM_PhyChan}. Then, \eqref{eq:af_sce_sc_1} becomes
\begin{align}
    A(\tau,\nu) &\approx e^{j2\pi\nu\tau} N\phi(-\nu NT) \sum_{\zeta} g\left(\tau+\zeta\frac{T}{M}\right)
    \nonumber\\
    &\ \times\sum_{m=0}^{M-1} e^{-j2\pi\nu\frac{T}{M}m} c[m] c^*\left[[m+\zeta]_{M}\right].
    \label{eq:af_sce_sc_2}
\end{align}
 As a special case of \eqref{eq:af_sce_sc_2}, the zero-Doppler cut $A(\tau,0) \approx M g\left(\tau\right)$ was given in \eqref{eq:A_ca_approx}. Similarly, the zero-delay cut can be obtained as $A(0,\nu) \approx N\phi(-\nu NT)$. Notably, $A(\tau,0)$ and $A(0,\nu)$ match those of the ODDM pulse \cite{Tong2024ODDM_PhyChan}, further highlighting the DD-SRN-FMCW waveform's consistency and suitability for ODDM-based ISAC systems.

\begin{figure}[t]
    \centering
    \includegraphics[width=0.9\columnwidth,trim={14 4 25 20},clip]{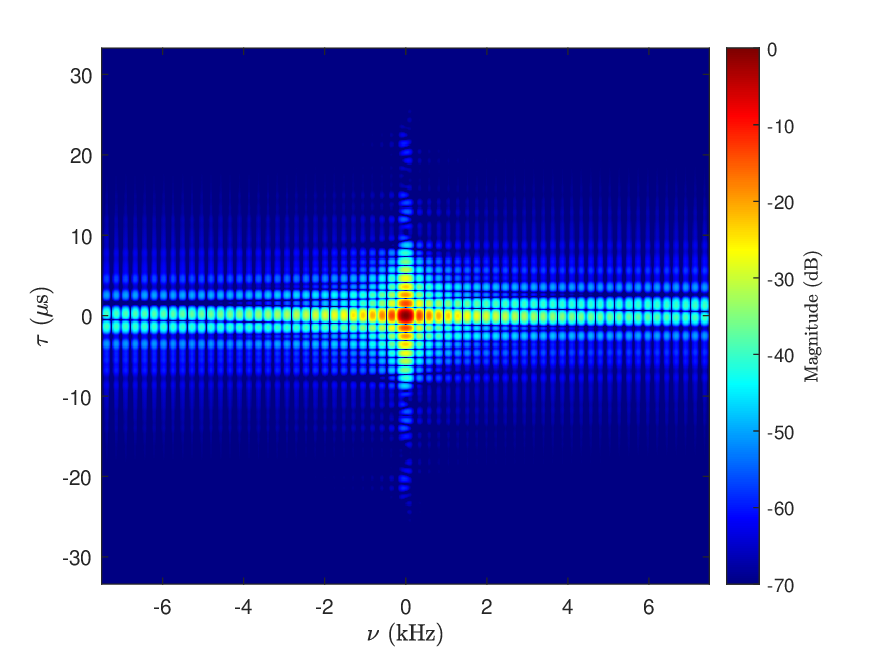}
    \caption{Ambiguity function $A(\tau,\nu)$ of the DD-SRN-FMCW signal ($M=N=64$).}
    \label{fig:af_all}
\end{figure}

\begin{figure*}[t]
    \centering
    \subfloat[DDIP.]{\label{fig:af_impulse}\includegraphics[width=0.32\textwidth,trim={14 4 25 20},clip]{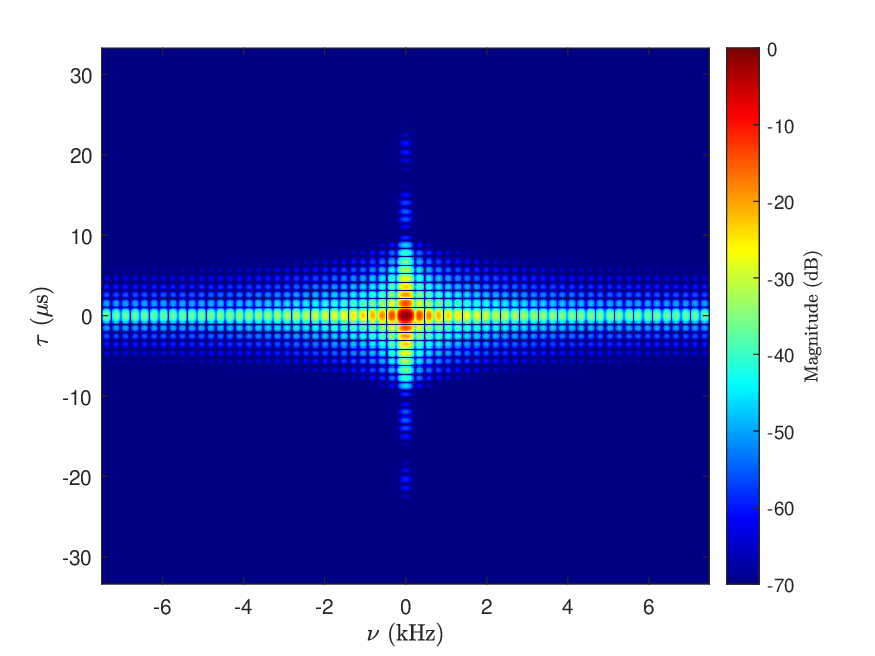}}
    \hfill
    \subfloat[Linear FMCW.]{\label{fig:af_fmcw}\includegraphics[width=0.32\textwidth,trim={14 4 25 20},clip]{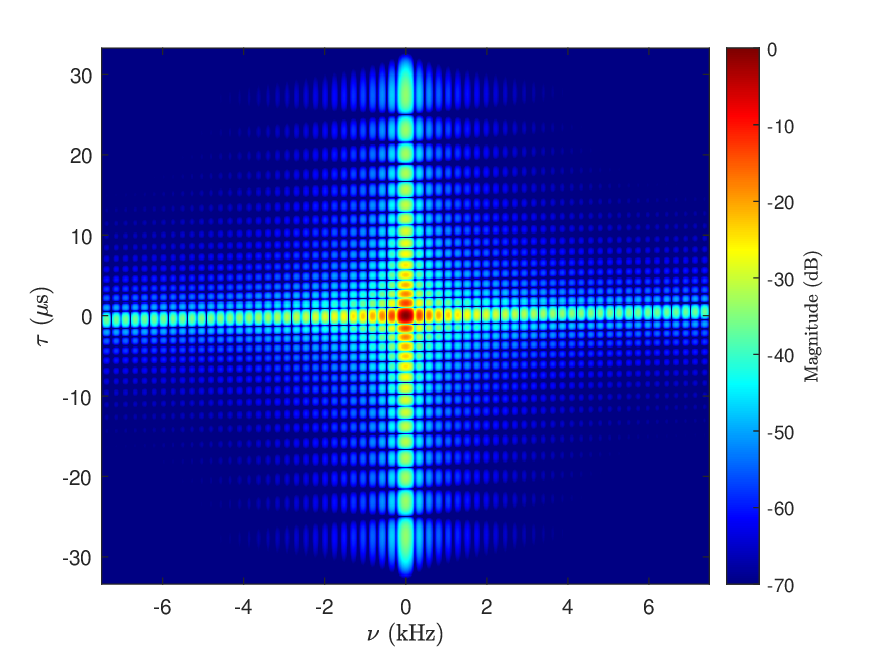}}
    \hfill
    \subfloat[2D chirp \cite{Ubadah2024ZakOTFS_ISAC}.]{\label{fig:af_chirp2d}\includegraphics[width=0.32\textwidth,trim={14 4 25 20},clip]{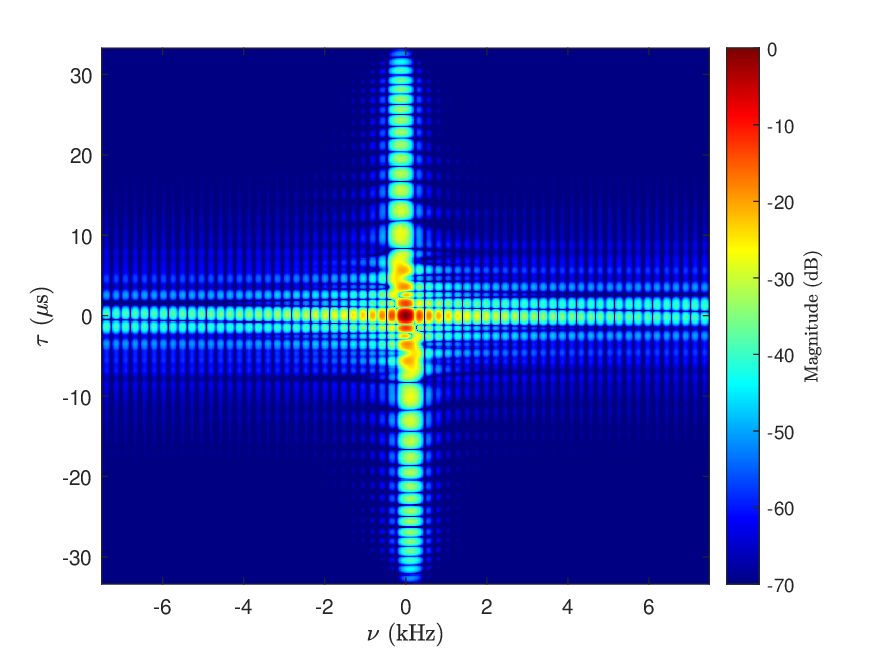}}
    \caption{Ambiguity function of different sensing signals ($M=N=64$).}
    \label{fig:af_comp}
\end{figure*}

For comparison, the ambiguity functions of three representative sensing signals are shown in Fig.~\ref{fig:af_comp}. The plot in Fig.~\ref{fig:af_comp}\subref{fig:af_impulse} shows the benchmark ambiguity function of a pulse-Doppler radar signal (also the DDIP in ODDM/OTFS \cite{Raviteja2019EmbeddedChannels,Tong2024ODDM_PhyChan}), which features uncoupled mainlobe and the lowest sidelobes. With a significantly lower PAPR, our proposed DD-SRN-FMCW signal preserves the ambiguity profile of DDIP with slightly higher off-diagonal sidelobes below \SI{-30}{\decibel}. By contrast, the linear FMCW signal in Fig.~\ref{fig:af_comp}\subref{fig:af_fmcw} also exhibits a good DD resolution, yet the 2D ambiguity plane is overlaid with a tilted lattice produced by range-Doppler coupling. In addition, as discussed in Section~\ref{sec:chirp_compression}, pronounced Fresnel-ripple sidelobes can be observed in the delay direction. To avoid coupling, the recently proposed 2D chirp signal \cite{Ubadah2024ZakOTFS_ISAC} leverages a baseband strategy similar to DD-SRN-FMCW. The ambiguity function of a 2D chirp signal is plotted in Fig.~\ref{fig:af_comp}\subref{fig:af_chirp2d}, where a slope parameter of $32$ is applied to achieve the desired mainlobe width. Although the 2D chirp signal has a low level of sidelobes in the Doppler direction, it exhibits a slanted sidelobe floor in the delay direction. In addition, it requires a significantly higher receiver complexity since a 2D correlator is required.

Compared to all three benchmarks in Fig.~\ref{fig:af_comp}, the proposed DD-SRN-FMCW waveform offers the most attractive compromise for an ISAC transceiver: it matches the low-PAPR advantage of FMCW, eliminates range-Doppler coupling, guarantees a deterministic low-sidelobe level, and remains fully compatible with baseband processing.

\subsection{Cramér-Rao Bound}\label{sec:crb}
Finally, we derive the CRB of the DD-SRN-FMCW signal to evaluate its superresolution sensing performance. We first derive CRB for the general ODDM frame $\boldsymbol{X}$. Denote the Fisher information matrix of $\boldsymbol{X}$ by $\boldsymbol{\mathcal{F}}(\boldsymbol{\theta,\boldsymbol{X}})\in\mathbb{R}^{4P\times 4P}$, where $\boldsymbol{\theta}=\left\{\abs{h_p},\angle h_p,l_p,k_p\middle|p=1,\dots,P\right\}$. The $i,j$-th element of $\boldsymbol{\mathcal{F}}(\boldsymbol{\theta,\boldsymbol{X}})$ is \cite{Gaudio2020ISACEffectiveness_OTFS}
\begin{align}
    \mathcal{F}_{i,j}(\boldsymbol{\theta,\!\boldsymbol{X}}){=}\frac{2}{\sigma_z^2}\Re\!\left\{\!\sum_{m=0}^{M-1}\!\sum_{n=0}^{N-1}\!\!\left(\!\frac{\partial Y[m,\!n]}{\partial\theta_{i}}\!\right)^{\!\!*}\!\frac{\partial Y[m,\!n]}{\partial\theta_{j}} \middle|\boldsymbol{\theta}\!\right\}.
\end{align}
Then, the MSE bound of the $i$-th parameter $\hat{\theta}_i$ is given by
\begin{align}\label{eq:crb}
    \mse\Bigl(\hat{\theta}_i\Bigr) \geq (\boldsymbol{\mathcal{F}}^{-1}(\boldsymbol{\theta,\boldsymbol{X}}))_{i,i}.
\end{align}

The partial derivatives w.r.t. $\abs{h_p}$ and $\angle h_p$ are readily given as $\frac{\partial Y[m,n]}{\partial \abs{h_p}} = e^{j\angle h_p} Y_p[m,n]$ and $\frac{\partial Y[m,n]}{\partial \angle h_p} = jh_p Y_p[m,n]$, respectively. $\frac{\partial Y[m,n]}{\partial l_p}$ and $\frac{\partial Y[m,n]}{\partial k_p}$ are more complicated with their closed-form espressions given in Appendix \ref{app:dldk}. The CRB for the DD-SRN-FMCW signal is then available by substituting $\boldsymbol{X}_c$ into \eqref{eq:crb}.

\begin{figure}[ht]
    \centering
    \subfloat[CRB of delay estimation $\hat{\tau}$.]{\label{fig:crb_tau}\includegraphics[width=0.5\columnwidth,trim={0 0 0 0},clip]{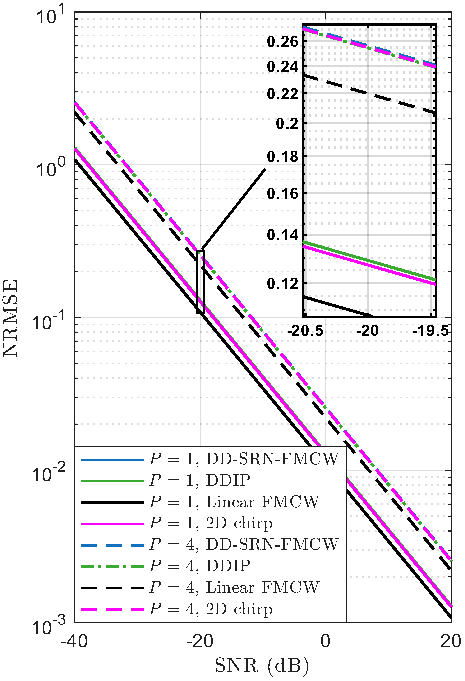}}
    \hfill
    \subfloat[CRB of Doppler estimation $\hat{\nu}$.]{\label{fig:crb_nu}\includegraphics[width=0.5\columnwidth,trim={0 0 0 0},clip]{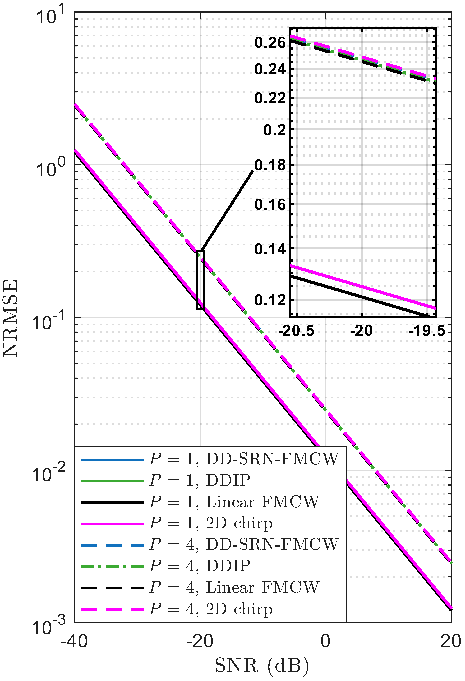}}
    \caption{CRB comparison for DD-SRN-FMCW, DDIP, linear FMCW, and 2D chirp.}
    \label{fig:crb_all}
\end{figure}

The CRBs for delay and Doppler estimation are plotted in NRMSE form and compared with those of other sensing signals in Fig.~\ref{fig:crb_all}. It is observed that the CRB of the DD-SRN-FMCW signal is comparable to that of other signals.

\section{Numerical Results}\label{sec:numerical}

In this section, we present the numerical results for the ISAC performance of the proposed ODDM-FMCW system. For the doubly-selective channel, we assume $P=4$ resolvable paths with uniformly distributed delay and Doppler shifts, where the maximum delay shift is \SI{33.36}{\micro\second} and the maximum Doppler shift is \SI{5003.46}{\hertz}. These values correspond to a maximum detectable range of \SI{5}{\kilo\meter} and a maximum radial velocity of \SI{150}{\metre/\second} (\SI{540}{\kilo\meter/\hour}) in a monostatic sensing configuration. For the transmitted signal, we consider 4-QAM modulation and a carrier frequency of $f_c=\SI{5}{\giga\hertz}$ with a nominal bandwidth of $\mathcal{B}=\frac{M}{T}= \SI{3.84}{\mega\hertz}$ and a frame duration of $NT=\SI{4.27}{\milli\second}$, based on $T=\SI{66.67}{\micro\second}$, $M=256$, and $N=64$. To reduce OOBE, a SRRC pulse is used as $a(t)$, having a roll-off factor of $\beta=0.15$. For the ISAC receivers, the DAS algorithm iterates $4$ times while the JCEDD algorithm iterates $8$ times. The internal OMP algorithm of both DAS and JCEDD algorithms performs $\mathcal{L}=10$ times grid evolution, yielding a superresolution factor of $2^{-\mathcal{L}}\approx0.001$.

We first present the CCDF of PAPR for DD-SRN-FMCW and ODDM-FMCW in Fig.~\ref{fig:ccdf_simulation}, compared with the DDIP signal \cite{Raviteja2019EmbeddedChannels} and ODDM with superimposed DDIP (ODDM-DDIP) \cite{Yuan2021DataAided_ChanEst}. For a fair comparison, no data is placed alongside the DDIP pilot on the $n=0$ Doppler column in ODDM-DDIP, matching the ODDM-FMCW frame structure in Fig.~\ref{fig:frame_oddmfmcw}. The energy of DDIP is also scaled to match the total energy of the corresponding DD-SRN-FMCW signal $\boldsymbol{X}_c$. DD-SRN-FMCW shows a superior PAPR cutoff value at around \SI{3}{\decibel}. Benefitting from the DD-SRN-FMCW signal component, the ODDM-FMCW signal enjoys a lower PAPR compared to the pure ODDM data signal. As $\rho$ increases, the DD-SRN-FMCW signal component becomes more dominant, leading to a further decrease in PAPR. In contrast, ODDM-DDIP exhibits a significantly higher PAPR, which is due to the DDIP component. As $\rho$ increases, the PAPR of the ODDM-DDIP signal increases rapidly, reaching beyond \SI{20}{\decibel} at $\rho=\SI{-5}{\decibel}$. Therefore, the ODDM-FMCW signal always achieves a significant PAPR reduction compared to the ODDM-DDIP signal.

\begin{figure}[t]
    \centering
    \includegraphics[width=0.95\columnwidth,trim={0 0 0 0},clip]{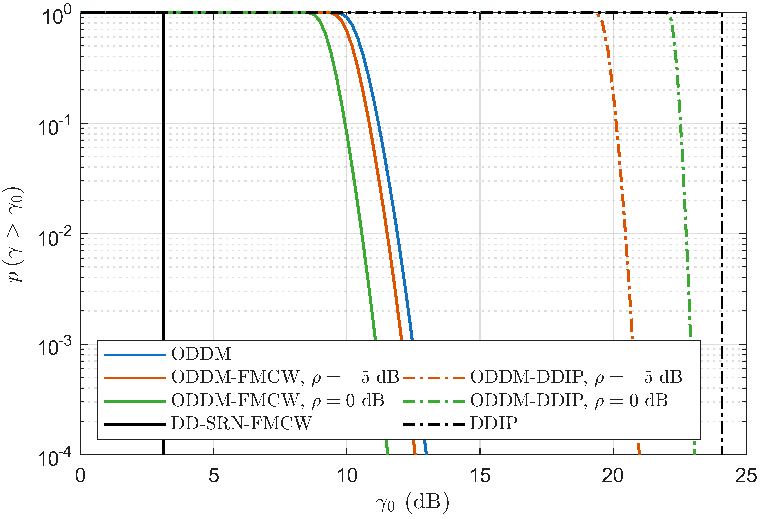}
    \caption{Simulated CCDF of different signals.}
    \label{fig:ccdf_simulation}
\end{figure}

In Fig.~\ref{fig:ber_vs_esn0}, we compare the BER performance of ODDM-FMCW and ODDM-DDIP at the communication receiver using the same JCEDD algorithm in Section~\ref{sec:jcedd}. We recall that JCEDD uses the known pilot component to estimate the channel and recover the data component. As a benchmark, ODDM with perfect CSI is also simulated using the soft SIC-MMSE detector \cite{Li2024SICMMSE_Turbo}. It is observed that ODDM-FMCW consistently outperforms ODDM-DDIP, achieving a more significant gain at high $E_s/N_0$ values. Moreover, with a CDPR of $\rho=\SI{-8}{\decibel}$, the BER performance of ODDM-FMCW with JCEDD closely approaches that of ODDM with perfect CSI.

\begin{figure}[t]
    \centering
    \includegraphics[width=0.95\columnwidth,trim={0 0 0 0},clip]{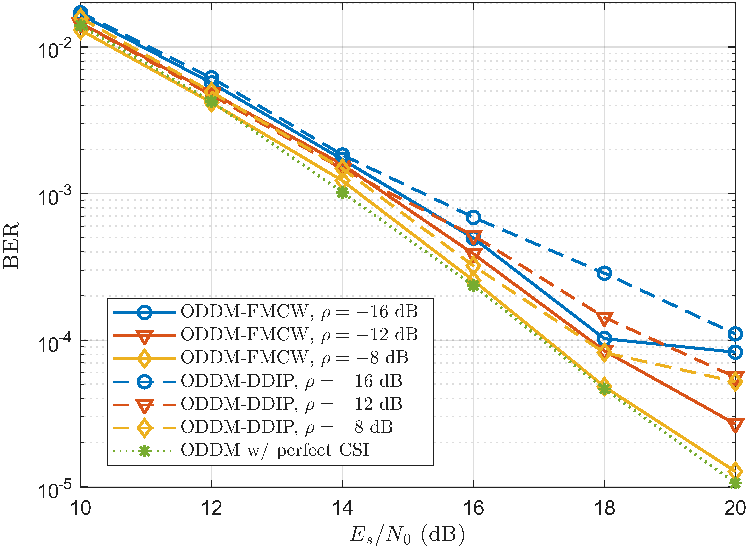}
    \caption{BER at the communication receiver. ODDM-FMCW and ODDM-DDIP use JCEDD; ODDM w/ perfect CSI uses soft SIC-MMSE \cite{Li2024SICMMSE_Turbo}.}
    \label{fig:ber_vs_esn0}
\end{figure}

To explain the BER gap between ODDM-FMCW and ODDM-DDIP, we examine the normalized mean square error (NMSE) of the channel estimation step within the JCEDD algorithm in Fig.~\ref{fig:nmse_jcedd}. Following \cite{Yuan2021DataAided_ChanEst}, we define the NMSE with respect to the equivalent response of a virtual unity-energy DDIP signal. The NMSEs obtained by transmitting pure DD-SRN-FMCW and DDIP serve as practical benchmarks, as they are unaffected by data-induced interference. Their corresponding CRBs (red solid and dashed curves) are also plotted for reference. Both DD-SRN-FMCW and DDIP achieve similar NMSE, consistent with the ambiguity function analysis in Sections \ref{sec:af_analysis} and the CRB analysis in \ref{sec:crb}. Moreover, the simulated NMSEs of DD-SRN-FMCW and DDIP closely track the CRBs, indicating the effectiveness of OMP with grid evolution for channel estimation. However, with superimposed data symbols, ODDM-DDIP exhibits a higher channel estimation NMSE than ODDM-FMCW. This is because the energy of the DDIP signal is concentrated in the DD domain, causing more severe interference to adjacent data symbols. The erroneous data detection negatively affects interference cancellation and, in turn, degrades channel estimation in subsequent JCEDD iterations. In contrast, the interference from the DD-SRN-FMCW pilot is distributed over the DD grid, thereby enhancing the robustness of ODDM-FMCW against channel estimation errors. This gives ODDM-FMCW a higher detection gain over ODDM-DDIP across the iterations of JCEDD.

\begin{figure}[t]
    \centering
    \includegraphics[width=0.95\columnwidth,trim={0 0 0 0},clip]{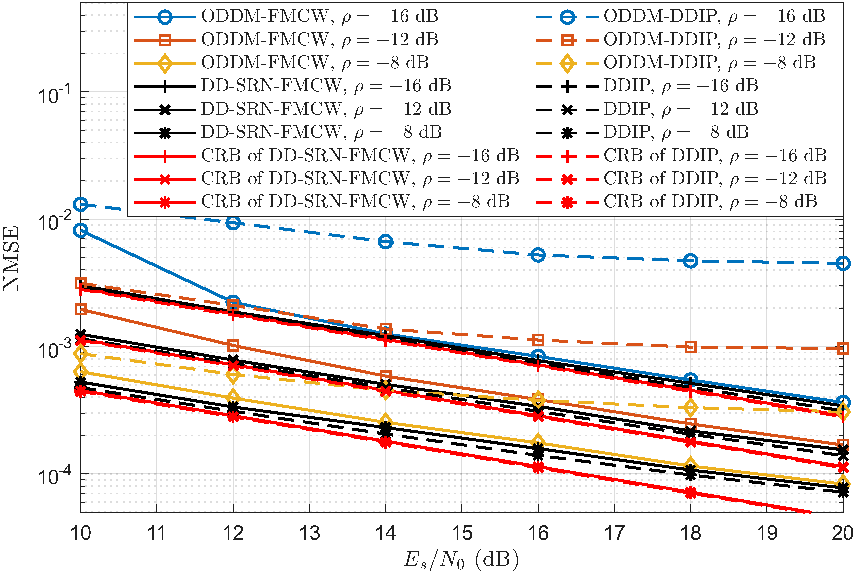}
    \caption{NMSE of channel estimation at the communication receiver. ODDM-FMCW and ODDM-DDIP use JCEDD; DD-SRN-FMCW and DDIP use OMP \cite{Rasheed2020OTFS_MU_OMP}.}
    \label{fig:nmse_jcedd}
\end{figure}

Then, we evaluate the NRMSE of delay and Doppler estimation at the collocated sensing receiver using the DAS algorithm in Figs. \ref{fig:rmse_das_tau} and \ref{fig:rmse_das_nu}, respectively. The NRMSEs of DD-SRN-FMCW and DDIP are also plotted as benchmarks, where the NRMSEs are obtained using the OMP algorithm \cite{Rasheed2020OTFS_MU_OMP}. The corresponding CRB curves are also plotted for reference. We recall that both the pilot and data components are known to the collocated sensing receiver. As expected, the NRMSE performance of DD-SRN-FMCW is comparable to that of DDIP, which is consistent with the NMSE results in Fig.~\ref{fig:nmse_jcedd}. However, once data symbols are superimposed and used as prior at sensing the receiver, ODDM-DDIP achieves a slightly lower NRMSE than ODDM-FMCW. This is because when the data are known, ODDM-DDIP no longer suffers from the detection gain penalty observed in Fig.~\ref{fig:ber_vs_esn0}. The residual data interference cancellation error now uniformly distributes over the DD domain, dependent only on the DD estimation quality. In this circumstance, the more DD-localized DDIP signal is only subject to error from adjacent DD samples, whereas the DD chirp compression used for ODDM-FMCW collects error energy from a wider DD region, making it more sensitive to imperfect DD estimates across successive DAS iterations. Despite this, both ODDM-FMCW and ODDM-DDIP approach the NRMSE performance of pure DD-SRN-FMCW and DDIP signals. In the low to medium $E_s/N_0$ regime, the NRMSEs obtained by OMP and by the proposed DAS algorithm closely follow their respective CRBs. As $E_s/N_0$ increases, however, the simulated NRMSE curves exhibit a slightly reduced slope. This behavior is mainly due to the finite DD resolution of the signals and the resulting high correlation among closely spaced paths, which becomes significant at high $E_s/N_0$.

\begin{figure}[t]
    \centering
    \includegraphics[width=0.95\columnwidth,trim={0 0 0 0},clip]{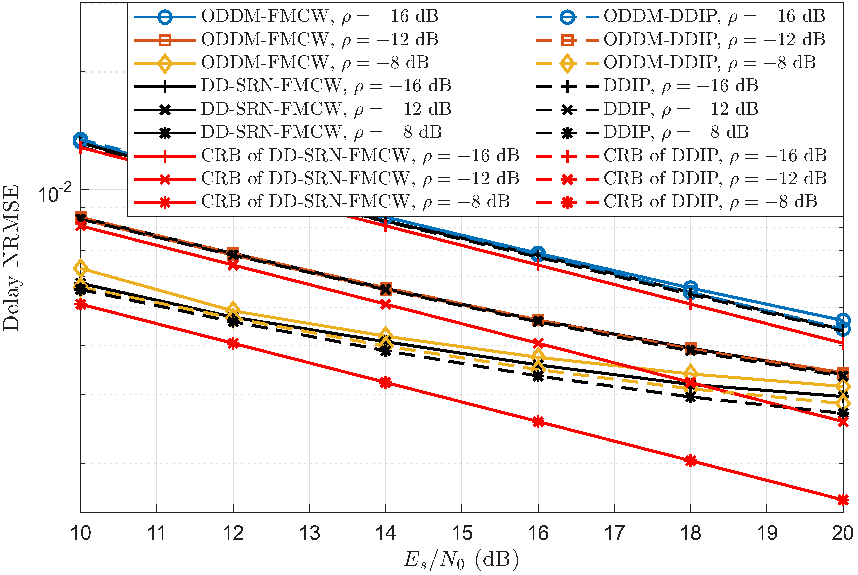}
    \caption{NRMSE of delay estimation at the sensing receiver. ODDM-FMCW and ODDM-DDIP use DAS; DD-SRN-FMCW and DDIP use OMP \cite{Rasheed2020OTFS_MU_OMP}.}
    \label{fig:rmse_das_tau}
\end{figure}

\begin{figure}[t]
    \centering
    \includegraphics[width=0.95\columnwidth,trim={0 0 0 0},clip]{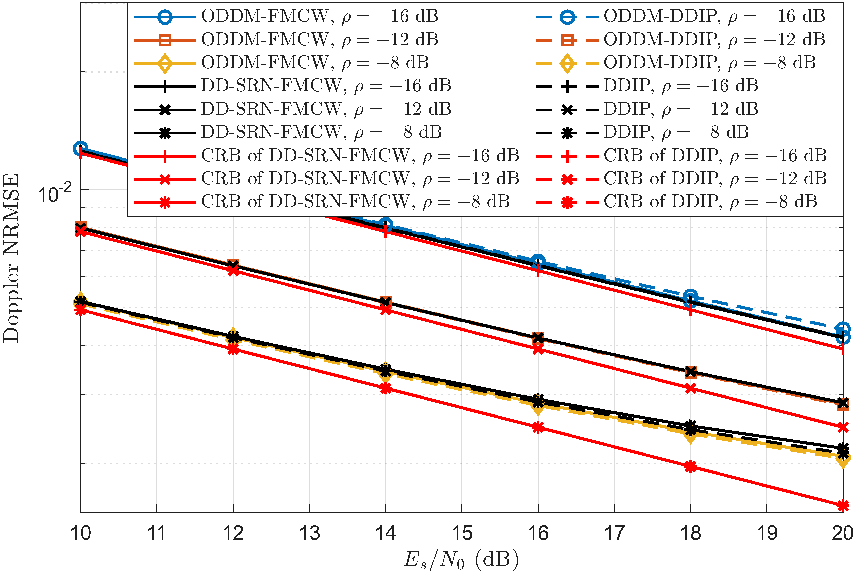}
    \caption{NRMSE of Doppler estimation at the sensing receiver. ODDM-FMCW and ODDM-DDIP use DAS; DD-SRN-FMCW and DDIP use OMP \cite{Rasheed2020OTFS_MU_OMP}.}
    \label{fig:rmse_das_nu}
\end{figure}

To further illustrate the sensing-communication tradeoff under a fixed total power constraint, we next consider a scenario where the total transmit SNR $(E_s+E_c)/N_0$ is kept constant while the CDPR $\rho$ is varied. In this case, increasing $\rho$ allocates more power to the DD-SRN-FMCW/DDIP component (and hence improves sensing), but reduces the power available for data symbols. Fig.~\ref{fig:ber_rho} shows the BER at the communication receiver versus $\rho$ for two total SNR levels, $(E_s + E_c)/N_0 = \SI{16}{\decibel}$ and $\SI{20}{\decibel}$. ODDM-FMCW and ODDM-DDIP both employ the JCEDD receiver, while ODDM with perfect CSI serves as a lower-bound benchmark using soft SIC-MMSE detection. For each total-power level, there exists an optimal range of $\rho\approx\SI{-10}{\decibel}$ where the BER is minimized. The reason is that, for low $\rho$, the pilot component is too weak to support accurate channel estimation and JCEDD suffers from error propagation; for large $\rho$, channel estimation improves but the reduced data power leads to a loss in detection performance. Similar to Fig.~\ref{fig:ber_vs_esn0}, ODDM-FMCW maintains a significant BER advantage over ODDM-DDIP across all $\rho$ thanks to the more distributed interference of DD-SRN-FMCW pilot to data symbols. Moreover, ODDM-FMCW can approach the perfect CSI benchmark at the optimal $\rho\approx\SI{-10}{\decibel}$, confirming its robustness under a fixed total-power constraint.

\begin{figure}[t]
    \centering
    \includegraphics[width=1\columnwidth,trim={0 0 0 0},clip]{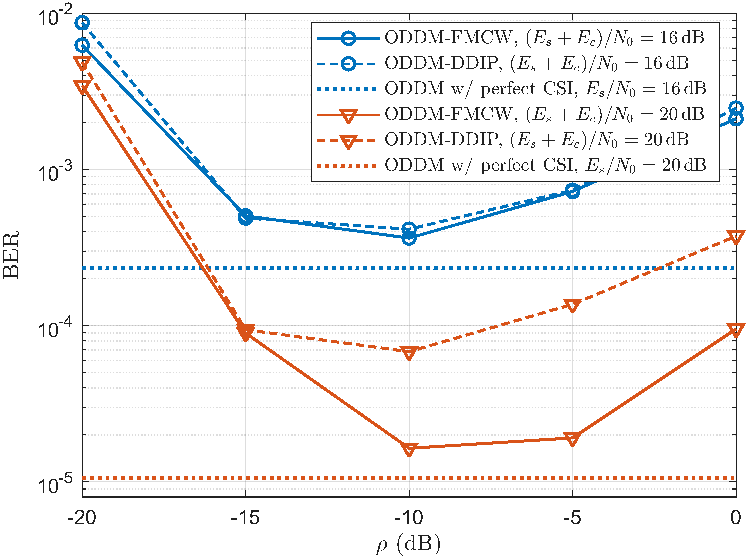}
    \caption{BER at the communication receiver with fixed total transmit power. ODDM-FMCW and ODDM-DDIP use JCEDD; ODDM w/ perfect CSI uses soft SIC-MMSE.}
    \label{fig:ber_rho}
\end{figure}

Figs.~\ref{fig:rmse_das_tau_rho} and \ref{fig:rmse_das_nu_rho} depict the corresponding delay and Doppler NRMSE at the collocated sensing receiver, again with fixed $(E_s + E_c)/N_0 = \SI{16}{\decibel}$ and $\SI{20}{\decibel}$. As expected, the NRMSE of both ODDM-FMCW and ODDM-DDIP decreases monotonically with $\rho$, since allocating more power to the pilot component directly enhances the effective SNR for parameter estimation. In addition, because the data symbols are known at the sensing receiver, ODDM-DDIP exhibits slightly lower NRMSE than ODDM-FMCW, consistent with the earlier discussion for Figs.~\ref{fig:rmse_das_tau} and \ref{fig:rmse_das_nu}. Overall, these results confirm that ODDM-FMCW can achieve a favorable sensing-communication tradeoff under a fixed total power constraint. In particular, $\rho\approx\SI{-10}{\decibel}$ provides BER close to the perfect CSI benchmark while maintaining a good sensing performance.

\begin{figure}[t]
    \centering
    \includegraphics[width=0.95\columnwidth,trim={0 0 0 0},clip]{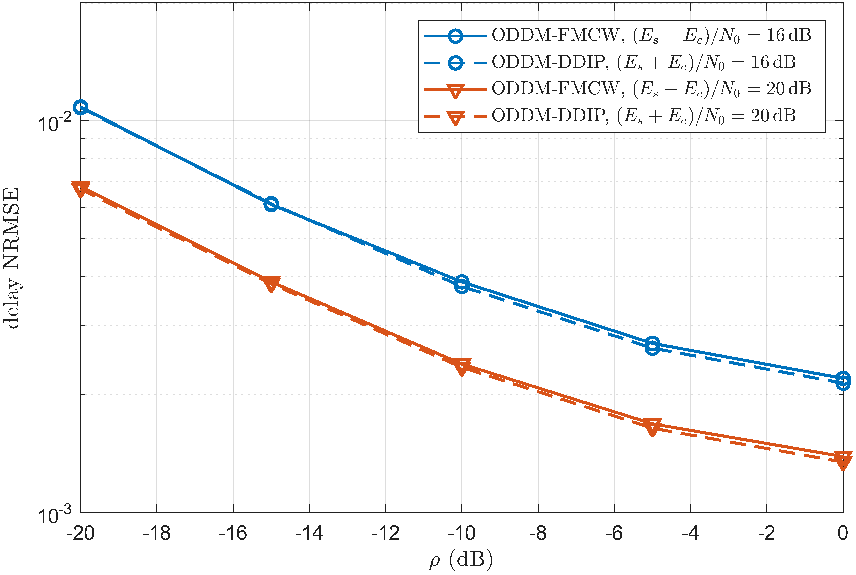}
    \caption{NRMSE of delay estimation at the sensing receiver using DAS.}
    \label{fig:rmse_das_tau_rho}
\end{figure}

\begin{figure}[t]
    \centering
    \includegraphics[width=0.95\columnwidth,trim={0 0 0 0},clip]{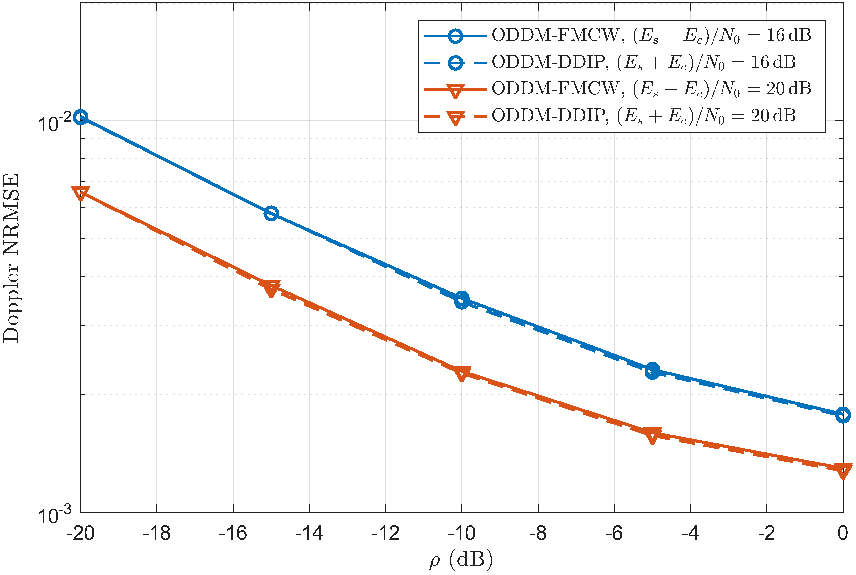}
    \caption{NRMSE of Doppler estimation at the sensing receiver using DAS.}
    \label{fig:rmse_das_nu_rho}
\end{figure}

\section{Conclusion}\label{sec:conclusion}

In this work, we presented ODDM-FMCW as a low-PAPR ISAC waveform for doubly-selective channels. We first proposed the DD-SRN-FMCW waveform to effectively integrate the FMCW waveform into the ODDM framework. Building on ODDM receiver operations, a DD chirp compression method is introduced to efficiently compute the DDR of the channel. DD-SRN-FMCW enjoys a high processing gain after DD chirp compression while circumventing the key drawbacks of conventional linear FMCW in ISAC applications.

Using the proposed DD-SRN-FMCW as a superimposed pilot for ODDM, we introduced the ODDM-FMCW waveform for ISAC, which can be directly transmitted by an ODDM transmitter. The high processing gain of DD-SRN-FMCW makes it easier to distinguish from random ODDM data symbols and facilitates interference cancellation methods. We further proposed a DAS algorithm for the collocated sensing receiver and a JCEDD algorithm for the communication receiver. We then presented a comprehensive analysis to prove the effectiveness of ODDM-FMCW, including the PAPR, spectrum, ambiguity function, and CRB. Numerical results show that ODDM-FMCW exhibits excellent ISAC performance in terms of BER for communications and NRMSE for sensing, while offering a significant reduction in PAPR compared to ODDM-DDIP. Moreover, using our JCEDD algorithm, ODDM-FMCW achieves lower BER than ODDM-DDIP, effectively approaching the BER performance of ODDM with perfect CSI. In general, ODDM-FMCW offers great potential for practical ISAC scenarios over doubly-selective channels.


\appendices

\section{Proof of Lemma \ref{lmm:zla}}\label{app:proof_zla}

We expand the LHS of \eqref{eq:zla} as
\begin{align}
    \sum_{m=0}^{M-1}\!\dot{c}[m]&\dot{c}^{*}[m\!+\!\zeta] = e^{-j2\pi\frac{\zeta}{M}\bigl(f_cT\!+\!\frac{\epsilon T^2\zeta}{2M}\bigr)}\!\sum_{m=0}^{M-1}\!e^{-j2\pi\frac{\epsilon T^2\zeta}{M^2}m}.
\end{align}
When $\frac{\epsilon T^2}{M}\in\mathbb{Z}$ and $\gcd\left(\frac{\epsilon T^2}{M},M\right)=1$, the inner summation simplifies to $\sum_{m=0}^{M-1} e^{-j2\pi\frac{\epsilon T^2\zeta}{M^2}m} = M\delta[\zeta]$. In addition, $e^{-j2\pi\frac{\zeta}{M}\bigl(f_cT+\frac{\epsilon T^2\zeta}{2M}\bigr)}=1$ when $\zeta=0$. These lead to the RHS of \eqref{eq:zla}, thus completing the proof.

\section{Proof of Proposition \ref{prop:zca}}\label{app:proof_zca}

Given Lemma \ref{lmm:zla}, $\frac{\epsilon T^2}{M}\in\mathbb{Z}$ and $\gcd\left(\frac{\epsilon T^2}{M},M\right)=1$ already ensure the ZLA property in \eqref{eq:zla}. To establish the ZCA property in \eqref{eq:zca}, it remains to show that the discrete chirp sequence $c[m]$ is periodic modulo $M$ if $f_cT+\frac{\epsilon T^2}{2}\in \mathbb{Z}$ is also satisfied, i.e., $c\bigl[[m]_{M}\bigr]=c[m]$. To verify this, we expand the LHS as
\begin{align}
    c\bigl[[m]_{M}\bigr]&=c[m+\kappa M]
    \nonumber\\
    &=c[m] e^{j2\pi\kappa\left(f_cT+\frac{\epsilon T^2}{2}\kappa\right)} e^{j2\pi\frac{\epsilon T^2}{M}m\kappa},
\end{align}
where $\kappa\in\mathbb{Z}$. When $\frac{\epsilon T^2}{M}\in\mathbb{Z}$, we have $e^{j2\pi\frac{\epsilon T^2}{M}m\kappa}=1$. In addition, $f_cT+\frac{\epsilon T^2}{2}\in \mathbb{Z}$ ensures that $e^{j2\pi\kappa\left(f_cT+\frac{\epsilon T^2}{2}\kappa\right)}=1$. Hence, the LHS becomes $c[m]$, which completes the proof.

\section{Derivation of $\frac{\partial Y[m,n]}{\partial l_p}$ and $\frac{\partial Y[m,n]}{\partial k_p}$}\label{app:dldk}
We first give expressions of two useful partial derivatives $\frac{\partial}{\partial l_p}g\!\left((d\!+\!\floor{l_p}\!-\!l_p)\frac{T}{M}\right)$ and $\frac{\partial}{\partial k_p}\phi(\tilde{n}+k_p-n)$, which characterize the ODDM pulse's sensitivity to $l_p$ and $k_p$, respectively. We consider $a(t)$ to be a SRRC pulse, so $g(t)$ is a raised-cosine pulse given by
\begin{align}\label{eq:gt}
    g(t) =
    \begin{cases} 
        \frac{\pi}{4}\sinc\left(\frac{1}{2\beta}\right), & \left|\bar{t}\right|=\frac{T}{2\beta M},
        \\
        \frac{\cos\left(\pi\beta\frac{M}{T}\bar{t}\right)}{1-\left(2\beta\frac{M}{T}\bar{t}\right)^2}\sinc\left(\frac{M\bar{t}}{T}\right), & \text{otherwise},
    \end{cases}
\end{align}
with $\bar{t}\triangleq t-Q\frac{T}{M}$. After some algebra, $\frac{\partial}{\partial l_p}g\!\left(\bar{l}\frac{T}{M}\right)$ can be expressed as \eqref{eq:dgdl}\addtocounter{equation}{1}, where we introduce the symbol $\bar{l}\triangleq d\!+\!\floor{l_p}\!-\!l_p-Q$ for brevity. Similarly, we introduce $\bar{k}=\tilde{n}\!+\!k_p\!-\!n$ and obtain
\begin{align}
    \frac{\partial}{\partial k_p}\phi\bigl(\bar{k}\bigr) = \frac{j2\pi}{N^2}\sum_{n'=0}^{N-1} n' e^{j2\pi\frac{n'\bar{k}}{N}}.
\end{align}

Finally, we can derive the derivatives of $Y[m,n]$ w.r.t. $l_p$ as in \eqref{eq:dY/dl} and w.r.t. $k_p$ as in \eqref{eq:dY/dk}.

\newcounter{MYtempeqncnt}
\begin{figure*}[!t]
    \normalsize
    \setcounter{MYtempeqncnt}{\value{equation}}
    \setcounter{equation}{52}
    \begin{align}\label{eq:dgdl}
        \frac{\partial}{\partial l_p}g\!\left(\bar{l}\frac{T}{M}\right)&=
        \begin{cases}
            0 & \bigl|\bar{l}\bigr| = \frac{1}{2\beta}
            \\
            \!\begin{aligned}
                \frac{\pi\beta \sinc\bigl(\bar{l}\bigr) \sin\bigl(\pi\beta \bar{l}\bigr)}{1-4 \beta^2 \bar{l}^2}\!-\!\frac{8 \beta^2\bar{l}\sinc\bigl(\bar{l}\bigr)\cos\bigl(\pi\beta\bar{l}\bigr)}{\left(1-4 \beta^2 \bar{l}^2\right)^2}\!+\!\frac{\pi\cos\bigl(\pi\beta\bar{l}\bigr)\biggl(\frac{\sin\left(\pi\bar{l}\right)}{\pi^2\bar{l}^2}\!-\!\frac{\cos\left(\pi\bar{l}\right)}{\pi\bar{l}}\biggr)}{1-4\beta^2\bar{l}^2}
            \end{aligned}
            & \text{otherwise}.
        \end{cases}
    \end{align}
    \setcounter{equation}{54}
    \begin{align}
        \frac{\partial Y[m,n]}{\partial l_p}&=h_p \sum_{d=0}^{2Q} e^{j2\pi\frac{(m\!-\!l_p\!-\!d)k_p}{MN}}\!\Biggl(\!\frac{\partial}{\partial l_p}g\!\left(\!\bar{l}\frac{T}{M}\!\right)\!-\!\frac{j2\pi k_p}{MN}g\!\left(\!\bar{l}\frac{T}{M}\!\right)\!\Biggr)\!\sum_{\tilde{n}=0}^{N-1}\phi\!\left(\bar{k}\right) \psi[m,d,\tilde{n}] X\!\left[\left[m\!-\!\floor{l_p}\!-\!d\right]_{M},\tilde{n}\right].
        \label{eq:dY/dl}\\
        \frac{\partial Y[m,n]}{\partial k_p} &= h_p \sum_{d=0}^{2Q} e^{j2\pi\frac{(m-l_p-d)k_p}{MN}} g\!\left(\!\bar{l}\frac{T}{M}\!\right) \sum_{\tilde{n}=0}^{N-1} \Biggl(\frac{\partial}{\partial k_p}\phi\left(\bar{k}\right)\!+\!\frac{j2\pi(m\!-\!l_p\!-\!d)}{MN}\phi\!\left(\bar{k}\right)\Biggr) \psi[m,d,\tilde{n}] X\!\left[\left[m\!-\!\floor{l_p}\!-\!d\right]_{M},\tilde{n}\right].
        \label{eq:dY/dk}
    \end{align}
    \hrulefill
    \vspace*{4pt} 

    \setcounter{equation}{\value{MYtempeqncnt}} 
\end{figure*}
\addtocounter{equation}{4}



\ifCLASSOPTIONcaptionsoff
  \newpage
\fi



\bibliographystyle{IEEEtran}
\bibliography{references}
%

\end{document}